\documentclass[twoside,11pt]{article}
\usepackage[abbrvbib]{jmlr2e}
\usepackage{def}
\usepackage{caption}
\usepackage{color}
\usepackage{subcaption}

\usepackage[T1]{fontenc}

\def\1{\bm{1}}

\def\min{\qopname\relax n{min}}
\def\max2{\qopname\relax n{max2}}
\def\max{\qopname\relax n{max}}

\def\argmax{\qopname\relax n{argmax}}

\newcommand{\RR}{\mathbb{R}}

\def\A{\mathcal{A}}

\def\D{\mathcal{D}}

\def\L{\mathcal{L}}

\def \cG {\mathcal{G}}

\def \cX {\mathcal{X}}

\def\a{\bm{a}} 
\def\x{\bm{x}} 
 
\def\y{\bm{y}}

\def\p{\bm{p}} 
 
\def\u{\bm{u}}

\newenvironment{lp*}{\begin{equation*}  \begin{array}{lll}}{\end{array}\end{equation*}}

\definecolor{darkgreen}{rgb}{0.09, 0.45, 0.27}

\jmlrheading{1}{2021}{1-48}{4/00}{10/00}{meila00a}{Jibang Wu, Haifeng Xu and Fan Yao}

\ShortHeadings{Uncoupled Multi-Agent Learning towards Rationalizability}{Uncoupled Multi-Agent Learning towards Rationalizability}
\firstpageno{1}

\newcommand{\fullmodel}{{Exp3 with Diminishing History}}
\newcommand{\model}{\texttt{Exp3-DH}}
\newcommand{\dir}{{diamond-in-the-rough}}
\newcommand{\game}{{diamond-in-the-rough game}}
\newcommand{\gameshort}{{DIR}}
\newcommand{\Algoclass}{{Merit-based}}
\newcommand{\algoclass}{{merit-based}}

\begin{document}

\title{Uncoupled Bandit Learning  towards Rationalizability: Benchmarks,  Barriers, and  Algorithms\thanks{The word ``towards'' indicates last-iterate convergence. A one-page abstract of this work, titled ``Multi-Agent Learning for Iterative Dominance Elimination: Formal Barriers and New Algorithms'' is accepted for presentation at the Conference on Learning Theory (COLT) 2022.}}

\author{\name Jibang Wu\thanks{Equal Contribution} \email wujibang@uchicago.edu \\
 \addr Department of Computer Science \\ University of Chicago \\ Chicago, IL 60637, USA
       \AND
       \name Haifeng Xu\footnotemark[2] \email haifengxu@uchicago.edu \\
       \addr Department of Computer Science \\ University of Chicago \\ Chicago, IL 60637, USA
       \AND
       \name Fan Yao\footnotemark[2] \email fy4bc@virginia.edu \\
       \addr Department of Computer Science \\ University of Virginia \\ Charlottesville, VA 22904, USA}


\maketitle

\begin{abstract}
 Under the uncoupled learning setup, the last-iterate convergence guarantee towards Nash equilibrium is shown to be impossible in many games. This work studies the last-iterate convergence guarantee in general games toward \emph{rationalizability}, a key solution concept in epistemic game theory that relaxes the stringent belief assumptions in both Nash and correlated equilibrium. 
This learning task naturally generalizes best arm identification problems, due to the intrinsic connections between rationalizable action profiles and the elimination of iteratively dominated actions. 
Despite a seemingly simple task, our first main result is a surprisingly negative one; that is, a large and natural class of no regret algorithms, including the entire family of Dual Averaging algorithms, provably take \emph{exponentially} many rounds to reach {rationalizability}. 
Moreover, algorithms with the stronger no swap regret also suffer similar exponential inefficiency. To overcome these barriers, we develop a new algorithm that adjusts \texttt{Exp3} with \texttt{D}iminishing \texttt{H}istorical rewards (termed \model{}); \model{} gradually ``forgets'' history at carefully tailored rates.  We prove that when all agents run \model{} (a.k.a., \emph{self-play} in   multi-agent learning), all iteratively dominated actions can be eliminated within polynomially many rounds. Our experimental results further demonstrate the efficiency of \model{},  and that state-of-the-art bandit algorithms, even those developed specifically for learning in games, fail to reach rationalizability efficiently.
\end{abstract}


\newpage 
\section{Introduction}
 Two seminal results in uncoupled~\footnote{In the uncoupled learning setup, the learning rule of each agent must not rely on any opponent’s historical actions or payoffs \cite{pradelski2012learning,daskalakis2011near,cai2023uncoupled}} multi-agent  learning are that agents using no regret learning algorithms will converge to a coarse correlated equilibrium (CCE) whereas the stronger no-swap regret learning algorithms will bring agents to a correlated equilibrium (CE) \cite{foster1999regret, blum2005external}. 
In both results, however, the converging sequence is the \emph{average} of agents' historical plays; the analysis of the \emph{last-iterate} convergence, a strictly stronger convergence guarantee, has been a well-known challenge in the study of uncoupled learning dynamics, despite significant research efforts devoted to its study even until today.~\footnote{Even in the zero-sum games, it is only recently shown by ~\citet{cai2023uncoupled} that a finite last-iterate convergence rate can be provably achieved by uncoupled learning dynamics under bandit feedback. See also a recent paper by \citet{anagnostides2022last} for detailed discussions about this challenge, relevant works and a few generalizations beyond zero-sum games in which last-iterate convergence could be possible.} 
For modern machine learning applications, the last-iterate convergence is often more desirable due to the difficulty of averaging agent’s actions, typically represented by neural networks. In addition,  it is also more realistic as an approach to predict equilibrium outcomes of economic systems, e.g., by assuming uncoupled learning agents in the system --- after all, it is the most recent  state (i.e., the ``last iterate'') of the system that matters and evolves, while the average history of the system may not have any real-world meaning.  Given these  challenges of obtaining last iterate convergence towards equilibria in general game as well as its strong  real-world motivations,  we seek to answer the following question in this paper:
\begin{quote}
    \it  Are there natural equilibrium concepts, other than the often studied CE or CCE, for which the last-iterate convergence guarantee of uncoupled learning dynamics can be established for general games? If so, what is the solution concept, and what algorithm is guaranteed to converge?   
\end{quote}   


 To answer the above question, we initiate the study of last-iterate convergence of uncoupled multi-agent bandit learning towards a fundamental game-theoretic solution concept, known as \emph{rationalizability}, developed through a series of seminal economic works~\cite{bernheim1984rationalizable,pearce1984rationalizable, brandenburger1987rationalizability,Milgrom1990rationalizability}.
In particular, this paper pursues \emph{rationalizability} as a natural multi-agent learning objective for several reasons. First, without converging to the rationalizability, it is impossible to reach the Nash equilibrium (NE) or CE (see Corollary~\ref{coro:supersetofce}). 
At a high level, rationalizability is a more permissive and robust solution concept that relaxes the stringent belief assumptions in both Nash \cite{bernheim1984rationalizable,pearce1984rationalizable} and CE \cite{aumann1987correlated,brandenburger1987rationalizability}:
while both NE and CE requires each player to respond optimally to the \emph{accurate} belief about her opponents' strategy profile, players under rationalizability may take any actions that are best responses to some \emph{erroneous but rational} belief about her opponents' strategy profile (see Definition \ref{def:rationalizable}).
This also leads to the second point that rationalizability is often deemed a more realistic outcome to expect in games with uncertainty~\cite{dekel2015epistemic}. The notion originates from the epistemic approach to understand agents' rational behavior in a non-cooperative environment without perfect knowledge of others' strategy profile --- it characterizes the outcomes arise from the only common knowledge of rationality. 
We thus can think of rationalizability as a proxy for us to understand how much ``rationality'' learning algorithms can obtain under the uncoupled learning setup, compared to human agents. 


Rationalizability is also  related to the basic strategic concept of \emph{dominance elimination} studied since the early days of the game theory field~\cite{gale1953theory,raiffa1957games} --- we say, an action $a$ of some agent is \emph{dominated} by another action $a'$ in a strategic game if the agent's payoff of action $a$ is always smaller than her payoff of $a'$, regardless of what actions other agents play. 
While it is debatable that whether regular humans  would ever play an equilibrium in a strategic game \cite{rabin1993incorporating,wright2010beyond}, it is widely observed and well accepted  that rational human players generally would avoid playing dominated actions \citep{fudenberg2019predicting}. Therefore, an intriguing question is whether there are learning algorithms that can efficiently approach such kind of human wisdom. 
Notice that, after eliminating some dominated actions, other actions may then start to become dominated and thus require an additional iteration of dominance elimination; the \emph{iterated dominance elimination} turns out to be a highly non-trivial task in the uncoupled learning setup and many of the existing algorithms are provably inefficient. As we will explain in Section \ref{sec:rationalizability}, the strategy profiles surviving this process of \emph{iterated dominance elimination} coincides with the set of rationalizable outcomes~\cite{bernheim1984rationalizable,aumann1987correlated}, and it thus serves another key motivation to design algorithms with efficient convergence guarantee towards rationalizability.


Notably, iterated dominance elimination in games is not as rare as they may first sound. For example, \citet{alon2021dominance} recently show that for randomly generated two-player $m\times n$ games with $m = o(\log (n))$, the fraction of actions that survive iterated dominance elimination tends to $0$ as $n\to\infty$. 
The concept also has a variety of applications, including voting \citep{moulin1979dominance}, auctions  \citep{azrieli2011dominance}, market design \citep{abreu1992virtual}, supermodular game \cite{Milgrom1990rationalizability}, oligopolistic competition \citep{borgers1995dominance} and global games \citep{carlsson1993global}. 
One celebrated example is Akerlof’s ``market for lemons'' \cite{akerlof1978market}. Each seller in this market is looking to sell used cars which are equally likely to have quality \texttt{H}/high, \texttt{M}/medium or \texttt{L}/low (low quality cars are also known as ``lemons'' in America). Prospective buyers value \texttt{H}-cars  at \$1000, \texttt{M}-cars at \$500 and \texttt{L}-cars at \$0, whereas
sellers value keeping a \texttt{H}-car at \$800, \texttt{M}-car at \$400 and \texttt{L}-car  at \$0 (these values are only privately known to sellers).  
Akerlof studies the situation that sellers precisely know their car's type whereas buyers cannot distinguish the good cars from lemons. Suppose that the car types are uniformly distributed on the market in the beginning, due to the inability to distinguish car quality, any buyer will immediately eliminate any price above her average value \$$500 =(1000+500+0)/3$. After this elimination, the buyer's   price becomes lower than   \texttt{H}-car's reservation value, and thus drive \texttt{H}-car sellers out of the market. Consequently, the buyer will \emph{gradually learn} that the market has no \texttt{H}-cars under price \$$500$ and thus will \emph{iteratively} eliminate any price above \$$250 =(500+0)/2$, which then further drives  \texttt{M}-car sellers out of the market. Ultimately, Akerlof observes that this iterated dominance elimination procedure will drive all good cars out of the market, and only lemons are ever traded. In this paper, we shall examine the more realistic situation when buyers and sellers do not know the exact average value of the car qualities in advance but only have noisy bandit feedback about each sold car. We seek to understand \emph{how fast the  market  collapse  observed by Akerlof may happen when players have such noisy bandit information feedback.} In Appendix~\ref{append:exp}, we will revisit this example as an empirical illustration of our theoretical results.


At this point, an immediate thought one might have is whether any standard no-regret learning algorithm would already suffice to eliminate the (obviously bad) dominated actions. The answer is indeed Yes, but with a crucial limitation that they may necessarily take \emph{exponentially} many rounds, as we will prove later in Section \ref{sec:barr}. A surprising insight revealed from our formal results is that the classic notion of \emph{regret} in multi-agent settings is not fully aligned with the performance of iterated dominance elimination. First, the history that standard no-regret algorithms exploit could become the inertia that impedes the iterative process of dominance elimination. This claim shall be self-evident in the proof of Theorem~\ref{thm:nonconverge}. Second, the notion of ``regret'', designed for either stochastic or adversarial settings, fails to encourage the coordination that facilitates the iterative learning process of the learning agents.  These observations echoes with the findings of \citet{viossat2013no} that the Hannan sets~\cite{hannan20164} may contain highly non-rationalizable outcomes.

Motivated by the aforementioned fundamentality and intricacies,  this paper studies how agents in a uncoupled multi-agent system can learn to \emph{rationalize} --- or equivalently to iteratively eliminate all dominated actions --- under \emph{noisy bandit information} feedback. This is also a natural generalization of the well-known \emph{action elimination} problem \cite{even2006action} to multi-agent setups.  Our study reveals interesting new challenges of learning in game-theoretical settings that its algorithm design may require different ideas from the classical online learning under either adversarial or stochastic environment assumptions. Notably, an interesting recent follow-up work by \citet{wang2022learning} considered the same problem setup as us and proposed an iterative best response approach that improves on our convergence rate by some linear factors.  However, their proposed algorithm is designed in a coordinated fashion: when an agent is estimating the mean reward of his actions, the remaining agents must keep playing the same action profile. {In such a design, each agent can  infer the strategies of all other agents at any time. While they obtained better learning efficiency than our algorithm,  such coordinated design of learning algorithm violates the uncoupled learning setup and is not our goal in this paper (also see additional discussions in Section \ref{sec:model}).}
In contrast, our work predicates the efficient convergence even when agents only have a loose agreement on the type of learning algorithms without any further intention or capability to coordinate (e.g., the sellers on the Akerlof's market for lemons). Moreover, experiments in Appendix \ref{sec:regret} demonstrate that the performance of our algorithm could remain robust facing potentially adversarial opponents. 

\paragraph{Contributions} At the conceptual level, our key contribution is to identify the solution concept of \emph{rationalizability} as a fundamental multi-agent learning objective. On the one hand, it is the best possible outcome under the uncertainty of others' strategy profile in uncoupled multi-agent learning problems; on the other hand, reaching rationalizability has important economic implications in various classes of games. 
Our technical contributions are twofold. First, we provide formal barriers of reaching \emph{rationalizability} under noisy bandit feedback. To do so, we identify an interesting benchmark class of dominance solvable game instances, coined \emph{\game{}} (DIR), and show that a broad class of no-regret online learning algorithms, including the Dual Averaging algorithm~\cite{nesterov2009primal, xiao2010dual}, has to run exponentially many rounds to eliminate all dominated actions with non-increasing learning rates. Moreover, we prove that the algorithms with the stronger no \emph{swap} regret suffers similar exponentially slow convergence. Second, we propose a new variant of the Exp3 algorithm with a carefully designed diminishing history mechanism to overcome the barriers in such learning tasks and prove its efficiency of eliminating all dominated actions within polynomially many rounds in the sense of last-iterate convergence. Our experiments demonstrate the effectiveness of our algorithm not only in the synthetic DIR games but also in other real-world games. 


\section{Additional Related Work} 
 
\paragraph{Rationalizability and Epistemic Game Theory}
While the classical game theory takes a top-down perspective, specifying the outcomes of a game by different solution concepts, the epistemic approach to game theory takes a bottom-up perspective, asking under what epistemic conditions will players behave with respect to particular solution concept~\cite{dekel2015epistemic}. In particular, using the mathematical tools developed in seminal works by Nobel laureates \citet{harsanyi1967games} and \citet{aumann1976agreeing}, it concerns decision problems under uncertainty such as the choices and knowledge of other players.  
The notion of \emph{rationalizability} is key to the development of epistemic game theory.  
\citet{bernheim1984rationalizable,pearce1984rationalizable} proposed the concept of rationalizability as the logical consequence of assuming that the only common knowledge is game structure and the rationality of the players, drawing an important connection to the iterated elimination of strictly dominated strategies. \citet{brandenburger1987rationalizability} introduced the solution concept of correlated rationalizability that allows players to have correlated conjectures over others’ actions. Notably, there has been increasing usage of the correlated version of rationalizability, in part based on the influential argument of \citet{aumann1987correlated}: in games with more than two players, correlation may express the fact that what player 3 thinks that player 1 will do may depend on what he thinks player 2 will do, which has no connection with any overt or even covert collusion between player 1 and 2. Our paper also adopts this notion of correlated rationalizability. 
In addition, \citet{aumann1987correlated} derives the CE from the additional assumption that the players are ``Bayesian rational'' with a common prior, and \citet{aumann1995epistemic} provides an epistemic characterization of the Nash equilibrium in terms of mutual knowledge of strategy choices. More recently, the solution concepts of interim independent rationalizability~\cite{ely2004hierarchies}, interim correlated rationalizability~\cite{dekel2007interim} are developed for the incomplete information games. 
 


\paragraph{Multi-agent Learning in Games}
Multi-agent learning in games has been of interest since the early days of artificial intelligence and economics~\cite{von2007theory, brown1951iterative, hart2000simple}. In recent years, there is a growing body of work on decentralized no-regret dynamics and their equilibrium convergence properties in various special classes of games including zero-sum game~\cite{daskalakis2011near, rakhlin2013optimization,syrgkanis2015fast, daskalakis2017training, daskalakis2018limit, mertikopoulos2018optimistic}, concave game~\cite{mertikopoulos2016learning, Bravo2018bandit, mertikopoulos2019learning}, potential games \cite{cohen2017learning}, monotone games \cite{NEURIPS2022_db2d2001}, and auctions \cite{feng2021convergence}. 
Different from the goal of these works on  convergence to CCE or NEs in special classes of   games, we target convergence in arbitrary multi-player games but to a relaxed equilibrium notion, i.e.,  learn to \emph{rationalize} by removing dominated actions. Indeed, economists~\cite{viossat2015evolutionary} framed this learning goal broadly into the question \emph{whether evolutionary processes lead economic or biological agents to behave as if they were rational.} 
Our study focuses on the convergence properties of various no-regret learning algorithms specifically in the multi-agent \emph{bandit} learning setting (a.k.a. the ``radically uncoupled'' setup~\cite{foster2006regret}). Interestingly, our proposed mechanism of diminishing history resonates with the well-established studies of both behavioral economy~\cite{fudenberg2016recency} and political science~\cite{axelrod1981evolution}. It is also seen in one form or another (such as increasing learning rate or recency bias) of many learning algorithms for different purposes~\cite{rakhlin2013optimization, syrgkanis2015fast, bubeck2017kernel, agarwal2017corralling, lee2020bias}, some of which even beyond the domain of online learning~\cite{jin2018q, brown2019solving}. 
In the experimental section, we will compare our algorithm with some of them from these previous works. 

Another line of work studied the fast convergence to approximate efficiency and CCEs by regularized learning algorithm with properties such as Variation in Utilities (RVU)~\cite{syrgkanis2015fast}, or low approximate regret~\cite{foster2016learning}. Different from the goal of these works on convergence to CCEs or NEs (for special game classes), we focus on a different but arguably equally fundamental goal, i.e.,  learning to \emph{rationalize} by removing iteratively dominated actions. Similar objective is also examined by \cite{viossat2013no}, who showed that continuous fictitious play (CFP) eliminates all strictly dominated strategies and therefore converges to the unique Nash equilibrium in strictly dominance solvable games. However, CFP needs the computation of agents' best response, which requires expert knowledge of the utility function. Meanwhile, \cite{cohen2017hedging} focused on the dominance elimination property of the no-regret learning algorithm, Hedge, on generic games. This algorithm  requires full information feedback and is thus not applicable to our setting where each agent can only observe the utility of actions she took and may not even know the existence of her opponents in the uncoupled learning setup. 

\paragraph{Optimism and Diminishing History}
The mechanism to focus on the recent experience is also seen in one form or another (such as increasing learning rate or recency bias) of many learning algorithms for different proposes. \cite{brown2019solving} introduced the variant of counterfactual regret minimization that discounts the prior iterations to prevent earlier costly mistake from hampering the convergence. \cite{jin2018q} used exponential discount for Q-learning but with respect to episode length instead of time. For online learning algorithms, \cite{rakhlin2013optimization} introduced an ``optimistic'' variant of Mirror Descent with a minor but effective modification that counts the last reward observation twice. \cite{syrgkanis2015fast} showed that natural classes of regularized learning algorithms with a property of recency bias provably achieves faster convergence rates in normal form games. \cite{chen2020hedging,daskalakis2021near,anagnostides2021near} shows that the optimistic variant of Hedge enjoys the $O(\operatorname{ploy}(\log T))$ regret, as well as the swap regret under the black-box transformation by \cite{blum2005external, stoltz2005internal}.  
In the online bandit learning setting, a helpful technique to introduce historical bias is to apply increasing learning rate \cite{bubeck2017kernel}. The increasing learning rate turns out to be  powerful in several recent works: \cite{agarwal2017corralling} used it to maintain a more delicate balance between exploiting and exploring so a master algorithm could perform almost as well as its best base algorithm, and \cite{lee2020bias} employed it to effectively cancel the potentially large variance of the unbiased estimators in high-probability regret bound analysis. However, as we will demonstrate in the experiment, none of these techniques can efficiently overcome the barrier of eliminating iteratively dominated actions, which necessitates our attempt of designing new algorithms. 

Interestingly, our mechanism of diminishing history resonates with the well-established studies of both behavioral economy and political science. 
\cite{fudenberg2016recency} pointed out that recency bias is a behavioral pattern commonly observed in game-theoretic environments.
In the renowned book, \textit{The Evolution of Cooperation}, \cite{axelrod1981evolution} empirically demonstrated that the tit-for-tat strategy (simply copying the last move of the opponent) is the most successful strategy in the repeated prisoner's dilemma game. They accordingly pointed out an important insight for game strategy design -- to be \emph{provocable} to both retaliation and forgiveness, as tit-for-tat ignores all the good or bad experience from the opponent's previous silence or betray more than one round ahead. 
Nevertheless, such aggressive strategy would not work in online learning, as memory is critical for algorithm to optimize its decision from the past observation. Hence, our proposed algorithm generalizes such philosophy with a carefully designed discounting mechanism to balance the influence of history in online decision making while maintain the ``provocability'' crucial in certain game theoretical environment. 

\section{Preliminaries}  \label{sec:model}
 
In this section, we introduce the problem setup of this paper, starting with some basic notations and definitions from game theory.  An $N$-player game in normal form consists of a (finite) set of agents $\cN = \{1, \dots, N \}$, where the $n$-th agent have a finite set of actions (or pure strategies) $\cA_n$.
Let $\cA := \prod_{n\in \cN} \cA_n$ denote the action space,   $\actions := (a_1, \dots  a_{N}) \in \cA $ denote the action profile, and $a_{-n} \in \cA_{-n}$ as the action profile excluding agent-$n$'s action. 
Without loss of generality, we assume every agent has $K$ actions, i.e., $|\cA_n|=K$.\footnote{As long as each player's number of actions is upper bounded by the constant $K$, our main results hold. In fact, our results only depend on the total number of actions across all players.}   Each agent $n$ has a payoff function $u_n: \cA  \to [-1,1]$ that maps the action profile $(a_n, a_{-n})$ of all agents' actions to the $n$th agent's payoff $u_n(a_n, a_{-n})$.\footnote{Bounded utility is assumed for convenience but not essential, since it can always be re-scaled. }  
We denote such game instance as $\cG := \cG(\cN, \cA, u)$. 
In addition, each agent $n$ may randomize her action by playing a \emph{mixed strategy}, $x_n \in \Delta_{\cA_n}$, from the simplex over $\cA_n$. 
Denote $\cX_n := \Delta_{\cA_n}$ as the mixed strategy space of agent $n$, and $\cX := \prod_{n\in \cN} \cX_n$ as the space of all mixed strategy profiles $x := (x_1, \dots, x_N)$ aggregating over all agents. Denote $x_n(a_n)$ as the probability of playing action $a_n$ under mixed strategy $x_n$. Let $u_n(x_n, x_{-n}) := \sum_{a_1 \in \cA_1} \cdots \sum_{a_N \in \cA_N} u_n(a_1, \dots,  a_{N}) \prod_{n\in \cN } x_n(a_n)$ be the expected payoff of agent-$n$ under the strategy profile $x$.



\paragraph{Dominated Actions and Iterated Dominance Elimination}
We say action $a_n$ is strictly \emph{dominated} by a \emph{mixed} strategy $x_n \in \Delta_{\cA_n}$, if $u_n( x_n, a_{-n}) > u_n( a_n, a_{-n})$, $\forall a_{-n}\in \cA_{-n}$.
Dominated actions are perhaps the simplest generalization of sub-optimal actions in single-agent decision making. 
The procedure for an agent to remove all her dominated actions is called \emph{dominance elimination}. In the single-agent setting, dominance elimination degenerates to the widely studied best arm identification since all other actions are dominated by the optimal arm~\citep{bubeck2009pure}. However, in multi-agent setups, eliminating all iteratively dominated actions
may require many \emph{iterations} of dominance elimination, as illustrated in the introduction.  Moreover, an action dominated by a mixed strategy is not necessarily dominated by any pure strategy. So it is important to consider dominance elimination by mixed strategies.
The process of iteratively applying such procedure to remove iteratively dominated actions is called  \emph{iterated elimination of strictly dominated strategies} (IESDS). 
This motivates our following natural definition of \emph{elimination length}.  

\begin{definition}\label{def:length}
For any finite game $\cG$, we define the \textbf{elimination length} $L_0$ as the minimum number of iterations that IESDS needs to eliminate all iteratively dominated actions in $\cG$. For any successful execution of IESDS with elimination length $L_0$, the corresponding \textbf{elimination path} is a sequence of \textbf{eliminated sets} $(E_l)_{l=1}^{L_0}$ where $E_l$ contains all eliminated actions until iteration $l \in \{1,\cdots, L_0\}$. 
\end{definition}
By definition, we have $E_1 \subset E_2 \subset \cdots \subset E_{L_0}$ and $|E_{L_0}|<\sum_{n=1}^N|\cA_n|= KN$. Since the elimination path $(E_l)_{l=1}^{L_0}$ of a game $\cG$ might not be unique, when we refer to an eliminated set $E_l$, we consider $E_l$ from \textbf{any} possible elimination path. Let $\Delta$ be the smallest utility gap between any iteratively dominated action and the correspondingly dominant strategy during IESDS over all possible elimination paths.~\footnote{This notation is to draw an analogy to the stochastic bandit setting, where $\Delta$ typically denotes the gap of the means to different reward distributions, which largely determines the intrinsic difficulty of the problem. With slight abuse of notation, we also use $\Delta$ to represent the simplex by convention; the two use cases should be easily distinguishable. }

\paragraph{Rationalizability and Rationalizable Actions}
Let $x_{-n}\in \cX_{\cA_{-n}}$ be a \emph{belief} of the $n$-th agent on the (possibly correlated) strategy profile of the other players.~\footnote{The original definition of (independent) rationalizability~\cite{bernheim1984rationalizable, pearce1984rationalizable} requires the \emph{belief} to be a product of independent probability measures on each of the action sets $\cA_{n'}$ for $n' \in \cN \setminus \{n\}$. Under this more restricted notion of rationalizability, the Theorem \ref{thm:ide-rationalizable} in the next section holds only in two-player games. However,  throughout this paper we will instead focus on the more commonly adopted definition of (correlated) rationalizability,~\cite{brandenburger1987rationalizability,Milgrom1990rationalizability}.  } 
We say an action $a_n \in \cA_n$ is the best response to a \emph{belief} $x_{-n}$, if $a_n \in \argmax_{a \in \cA_n} u_n(a, x_{-n}) $. The notion of rationalizable action is defined in a recursive way:  an action is rationalizable if it is the best response to a ``rational'' belief supported on other agents' rationalizable actions. This situation happens under the common knowledge of rationality (despite incomplete knowledge of other agents' strategy) --- that is, every agent is rational, every agent thinks that every agent is rational, every agent thinks that every agent thinks that every agent is rational, and so on in any higher order beliefs. 
We formalize it in the following definition.

\begin{definition}[Rationalizable Actions \cite{osborne1994course}] \label{def:rationalizable}
Suppose \\that there exists $\{\cZ_n \subseteq \cA_n\}_{n=1}^{|\cN|}$ such that for any $n'\in \cN, a_{n'} \in \cZ_{n'}$, $a_{n'}$ is a best response to some belief supported only on $\cZ_{-n'}$. 
In this case, we say an action $a_n \in \cA_{n}$ is rationalizable for agent $n$. The solution concept of \emph{rationalizability} is formed by \emph{any} strategy profile that supports only on the rationalizable actions. 
\end{definition}

\paragraph{Problem Setup} We study how multiple agents can learn to rationalize 
under noisy bandit payoff feedback in an radically uncoupled fashion. At each round $t\in [T]$, each agent $n$ takes an action $a_n(t)$, which together forms the action profile $\actions(t)$. Then, agent $n$ individually observes from the environment a noisy, bandit feedback,  $u_n(\actions(t)) + \xi_{n,t}$, that is, its payoff under the action profile $\actions(t)$ perturbed by a noise $\xi_{n,t}$. With a slight abuse of notation, we denote $u_{n}(t) =u_n(\actions(t)) + \xi_{n,t}$ hereinafter. We assume $\{\xi_{n,t},\mathcal{F}_t\}_0^{+\infty}$ is a Martingale Difference Sequence (MDS) with finite variance. Specifically, $\{\xi_{n,t},\mathcal{F}_t\}_0^{+\infty}$ satisfies: 1) Zero-mean: $\mathbb{E}[\xi_{n,t}|\mathcal{F}_{t-1}]=0 \textup{ for all } t = 1, 2, \dots (a.s.)$; 2) Finite variance: $\exists \sigma > 0$ s.t.\ $\mathbb{E}[\|\xi_{n,t}\|^2|\mathcal{F}_{t-1}] \leq \sigma^2 \textup{ for all } t = 1, 2, \dots (a.s.)$. 
 
 
\paragraph{Learning Goals}
Given the above problem setting, we focus on designing an online learning algorithm that when all agents use such an algorithm, they will learn to play rationalized strategy with high probability in a last-iterate convergence manner.  
We require that each agent's learning rule is radically (or strongly) uncoupled~\cite{hart2003uncoupled, foster2006regret} 
in the sense that an agent's strategy has to be adaptive to her own historical observations and meanwhile cannot depend on all other agents’ historical actions and payoffs. Intuitively, the motivation of  uncoupled learning is to insure that any agent's strategy cannot be inferred by other agents. This crucially requires each agent's strategy to depend on his own information, even conditioned on other agents' information.  
Note that the coordination-based algorithm of \cite{wang2022learning} is \emph{not} uncoupled because an agent's strategy can be inferred from others' strategy profile at any time according the coordination rule specified by their algorithms.
Conceptually similar requirement is imposed in many works on multi-agent learning~\footnote{See \cite{daskalakis2011near, farina2022near, anagnostides2022uncoupled, cai2023uncoupled}.} to avoid the degeneracy to tailored learning rules of computing approximated equilibria from estimated game payoffs.
We also remark that many of the prior works~\footnote{See \cite{hannan20164,freund1999adaptive,daskalakis2011near, cherukuri2017saddle, cohen2017hedging, syrgkanis2015fast, mertikopoulos2019learning, mazumdar2020gradient}.} in multi-agent learning either focus on full information or gradient feedback, or consider the time-average convergence, our learning objective of last-iterate convergence under noisy bandit feedback provides arguably the strongest and most stable guarantee. 

 


\section{On the Pursuit of Rationalizability}\label{sec:rationalizability}
The notion of rationalizability and dominance elimination has been less popular to the machine learning community than  other   game-theoretical solution concepts such as NE and CE.
Thus to motivate our seemingly ``unconventional'' pursuit of the rationalizability, we devote this  section to highlight its theoretical and conceptual importance, as well as its potential real-world applications.

\subsection{Key Properties of Rationalizability}
Rationalizability turns out to be equivalent to iterated dominance elimination due to the elegant minimax theorem, as formally described by the following theorem.
\begin{theorem}\cite{osborne1994course} \label{thm:ide-rationalizable}
The set of any agent $n$'s rationalizable actions are precisely the set of agent $n$'s  actions that survive the process of IESDS. 
\end{theorem}
The proof of the above theorem hinges on the following observation. That is, any agent $n$'s action $\hat{a}_n$ is a strictly dominated action if and only if there exists \emph{no} belief $\mu_n$ over all other agent's actions that makes $\hat{a}_n$  a best response of agent $n$ (and thus $\hat{a}_n$ cannot be rationalizable).  
To see this, note that $\hat{a}_n$ is  a strictly dominated action implies there exists some mixed strategy $x_n$ such that $ u_n(x_n, a_{-n})- u_n(\hat{a}_n, a_{-n}) > 0 $ for any $a_{-n}$. Equivalently, $\max_{x_n \in \Delta_{\cA_n} } \min_{a_{-n}} [u_n(x_n, a_{-n})- u_n(\hat{a}_n, a_{-n}) ] > 0 $.  By the linearity of $u_n(x_n, a_{-n})$ in $x_n$ and strong duality, we know  $ \min_{\mu_n \in \Delta_{\cA_{-n}}} \max_{a_n \in \cA_n} [u_n(a_n, \mu_n)- u_n(\hat{a}_n, \mu_n) ] > 0 $. That is, for any agent $n$'s belief $\mu_n \in \Delta_{\cA_{-n}}$, there is an action $a_n$ that is strictly better than $\hat{a}_n$ and thus $\hat{a}_n$ can never be a best response to any agent $n$'s belief. 

From the perspective of epistemic game theory, rationalizability is a solution concept that only assumes the common knowledge of rationality, while
CE additionally assumes the common prior, i.e., a correct belief of other players' strategy profile~\cite{brandenburger1987rationalizability, aumann1987correlated}. These observations lead to an important fact that the set of rationalizability is a superset of correlated equilibria.

\begin{corollary}\cite{osborne1994course}\label{coro:ce-rationalizable}
For any finite game, every action used with positive probability in a CE is rationalizable.
\end{corollary}
This is also implied by Theorem \ref{thm:ide-rationalizable}, as no CE should put non-zero probability on any iteratively dominated action profile and thus must be rationalizable. 

In addition, if the process of IESDS terminates with a single action profile in the remaining action space, this action profile must be the unique NE and also the unique CE of the game \citep{viossat2008having}. In this case, game $\cG$ is called \emph{mixed-strategy solvable} \cite{alon2021dominance} --- a more general notion than the classic dominance solvable game that is defined on dominance elimination by pure strategy.   
\begin{corollary}\label{coro:supersetofce}
For any mixed-strategy solvable game, the only rationalizable action profile coincides with the unique NE (thus the unique CE) of the game. 
\end{corollary}

Rationalizability also has nice predictions for a large class of economic games, known as the supermodular games~\cite{topkis1979equilibrium, Milgrom1990rationalizability}, {which encompass many well-known games such as Cournot duopoly \cite{cournot1838recherches} and Bertrand competition \cite{bertrand1883review}}. Generally speaking, a finite, normal-form game is supermodular if each agent $n$'s utility function $u_n(a_n, a_{-n})$ has \emph{increasing difference} in $a_n$ and $a_{-n}$, i.e., for all $a'_n \geq a_n$ and $a'_{-n} \geq a_{-n}$, $u_n(a'_n, a'_{-n}) - u_n(a_n, a'_{-n}) \geq u_n(a'_n, a_{-n}) - u_n(a_n, a_{-n}) $. This require the set $\cA_n, \cA_{-n}$ to be partially ordered, e.g., as the price or effort level.~\footnote{More generally, when players have multi-dimensional strategy spaces, $\cA_n$ must be a complete lattice and $u_n$ is supermodular in $a_n$ for any fixed $a_{-n}$. When $u_n$ is twicely differentiable, it is equivalent to have $\partial^2 u_n/\partial a_{n} \partial a_{m} \geq 0, \forall m\neq n$. } 
The formal definition of supermodular games is deferred to Appendix \ref{append:supermodular}. One example of supermodular game is the bank run model~\cite{diamond1983bank}, when more depositors withdraw their funds from a bank, it is better for other depositors to do the same (see Section \ref{sec:prelim:example} for other examples).

\begin{theorem}\cite{Milgrom1990rationalizability}
For any supermodular game, the largest (resp. smallest) rationalizable strategy profile coincides with the largest (resp. smallest) NE of the game.
\end{theorem} 
The proof of the above theorem uses a fixed point argument: there exists a function $f: \cA \to \cA$ such that, from the largest action profile $\a^0$ in $\cA$, iteratively setting $\a^{t+1} = f(\a^t)$ leads to a fix point, $\a^{*} = \lim_{i\to \infty} \a^t$, which is also the largest NE of the game. Specifically, we can construct $f(a_1, a_2, \cdots, a_{-N}) = (\overline{\beta_i}(a_{-1}), \overline{\beta_i}(a_{-2}), \cdots, \overline{\beta_i}(a_{N}))$, where  $\overline{\beta_i}(a_{-n})$ denotes the largest element in agent $n$'s best response set, $\argmax_{a\in \cA} u_n(a, a_{-n}) $.  The sequence $\{ \a^t \}$ is non-increasing in $t$ by induction: Since $\a^0$ is the largest action profile, $\a^1 \leq \a^0$. Given that $\a^{t} \leq \a^{t-1}$, we have for any $n\in \cN$, $a_n^{t+1} = \overline{\beta_i}(a_{-n}^{t}) \leq \overline{\beta_i}(a_{-n}^{t-1}) = a_n^{t}$, since $\forall a_{n}\in \cA_{n}, u_n(a_n, a_{-n}^{t}) - u_n(a_n^{t+1}, a_{-n}^{t}) \leq u_n(a_n, a_{-n}^{t-1}) - u_n(a_n^t, a_{-n}^{t-1}) $ by supermodularity. The fixed point $\a^{*} = f(\a^*)$ is an NE by definition, as $a_n^* = \overline{\beta_i}(a_{-n}^*), \forall n\in \cN$. 

Meanwhile, $\a^{*}$ is also the largest rationalizable action profile, as any action $a_n > a_n^*$ is iteratively dominated. That is, for any agent $n$, any of its action $a_n > a_n^1 $ is dominated, since $u_n(a_n, a_{-n}) - u_n(a_n^1, a_{-n}) \leq u_n(a_n, a_{-n}^0) - u_n(a_n^0, a_{-n}^0) < 0$. By induction, given that any $a_{-n} > a_{-n}^{t-1}$ is (iteratively) dominated, any action $a_n > a_n^t = \overline{\beta_i}(a_{-n}^{t-1})$ is iteratively dominated: by supermodularity, $\forall a_{-n} \leq a_{-n}^{t-1}, u_n(a_n, a_{-n}) - u_n(a_n^t, a_{-n}) \leq u_n(a_n, a_{-n}^{t-1}) - u_n(a_n^t, a_{-n}^{t-1}) < 0 $, where the last inequality is strict as $a_n^t \notin \argmax_{a\in \cA} u_n(a, a_{-n}^{t-1})$.

Similarly, we can show that starting from the smallest action profile and iteratively taking the smallest best response mapping $\underline{\beta_i}(\cdot)$ would lead to the smallest NE. This proof at its core is built upon two seminal results: Tarski's fixed point theorem \cite{tarski1955lattice} that for any monotone function on a complete lattice, the set of all its fix points forms a sublattice; Topkis' monotonicity theorem \cite{topkis1979equilibrium} that for any supermodular game, each agent's best response function is monotone.



\subsection{Notable Examples}\label{sec:prelim:example}


Here we list several well-known game instances where rationalizability are desirable outcomes with important economics implications. The first two classes of games are dominance solvable games, and the other two are supermodular games.

\paragraph{The Market for ``Lemons'' \cite{akerlof1978market}} \label{sec:lemon}
Consider a market of used cars with a buyer and $N$ sellers. Each seller $i$ has a car of quality $q_i$, and two actions $a_i\in \{1, 0\}$, respectively, to list or not to list his car. Without loss of generality, let $q_N > q_{N-1} > \cdots > q_1$.  The buyer has his action $p \in \mathcal{P}$ from a set of prices  to buy a car from sellers. Suppose the buyer and sellers move simultaneously. As assumed by Akerlof, the buyer has no information about each seller's car quality before posting the price. The seller also decides whether or not to list his car (i.e., $a_i = 1$ or $0$) without knowing the price. In our experiments, we assume seller $i$ has a reservation value  $\tilde{q}_i= q_i + \epsilon$  which is  a noisy perception of his car quality with zero-mean noise $\epsilon$.  For those who did choose to list their cars, if the buyer's price are below their reservation value, they would refuse to sell, but still suffers a small and fixed   opportunity cost $c_1$. We denote $b_i=\one[\tilde{q}_i \leq p]$ as whether the car of seller $i$ gets sold. In contrast, the buyer is uninformed of the car quality he could buy, though the trade would generate welfare so that her revenue is a multiplier $c_2$ of the average quality $\bar{q} = \frac{\sum_{i\in [N]}b_i\cdot q_i}{\sum_{i\in [N]} b_i}$. Hence, each seller $i$ receives payoff $u_i = a_i ( b_i(p - q_i) - c_1)$, and the buyer receives payoff $u_0 = c_2 \cdot \bar{q} - p  $, where $c_1>0, c_2>1$ are game parameters. 

Our Proposition \ref{prop:lemon} in Appendix \ref{appendix:sec:lemon} shows that this game has elimination length at least $2N-1$ for small $c_1$ and its NEs are that the buyer sets a price no higher than $q_1$, and all sellers refuse to list. This outcome is certainly not a surprise and was observed by Akerlof as a \emph{market collapse}, due to asymmetric information among sellers and buyers  which gradually drives the high quality car sellers out of the market in order. The additional insight of our Proposition \ref{prop:lemon} is that this market collapse may follow a long procedure of iterated dominance elimination. We remark that the opportunity cost $c_1>0$ is only assumed to capture how much  sellers prefer  not selling if the price just matches their values. Our results hold for $c_1 = 0$ as well, but will need to work with weakly dominance elimination with a tie breaking in favor of not listing. 


\paragraph{Decentralized Matching under Aligned Preference \cite{niederle2009decentralized}}  In a decentralized matching market game, each firm (on one side) simultaneously make an offer to a single worker (on the other side), and the workers accordingly decides whether to accept, reject or defer the offers they receive.  \citet{ferdowsian2020decentralized} show that if the firms and workers' preferences are aligned (i.e., firms ranks the workers identically and vice versa), then the outcome of stable matching in this market is rationalizable, and moreover, the unique NE surviving iterated elimination of weakly dominated strategies. This rationalizable strategy profile coincides with the celebrated deferred acceptance algorithm by \citet{gale1962college}. 


\paragraph{Bertrand Competitions \cite{bertrand1883review}} This is a fundamental economic model of competition . It considers a market of $N$ firms that produces identical goods (i.e., substitutes). Each firm $n$ produces with constant unit costs $c_n$ and sets a price $a_n$ from some (possibly discretized) interval $[0, a_{\max}]$. Each firm has a utility function $u_n(a_n, a_{-n}) = (a_n - c_n) \cdot D_n(a_n, a_{-n})$, where $D_n(a_n, a_{-n})$ is the demand function with its elasticity being a non-increasing function of the other firms' prices $a_{-n}$. \citet{Milgrom1990rationalizability} shows that many common forms of demand such as linear function, logit function satisfies this property and the corresponding Bertrand competitions are therefore supermodular games.

\paragraph{Arms Races \cite{Milgrom1990rationalizability}} In this game, the players are two countries engaged in an arms race. Each player chooses a level of arms $a_n$ from some (possibly discretized) interval $[0, a_{\max}]$ and receives as its payoff $u_n(a_n, a_{-n}) = - C(a_n) + B(a_n - a_{-n})$, where $C$ is the cost function based on the arm level and $B$ is the welfare function which is concave w.r.t. the difference of arm level. That is, the marginal return to additional arms  is an increasing function of the foe's armament level. This is a supermodular game, since for all $a'_n \geq a_n$ and $a'_{-n} \geq a_{-n}$, $u_n(a'_n, a'_{-n}) - u_n(a_n, a'_{-n}) \geq u_n(a'_n, a_{-n}) - u_n(a_n, a_{-n}) \iff B(a'_n - a'_{-n}) - B(a_n - a'_{-n})  \geq B(a'_n - a_{-n}) - B(a_n - a_{-n}) $, due to the concavity of $B$.

\section{Formal Barriers of Multi-Agent Learning towards Rationalizability} \label{sec:barr}
Perhaps surprisingly, we show that even for dominance solvable games under full information feedback, standard bandit learning algorithms will necessarily take \emph{exponentially} many rounds to eliminate all dominated actions. Since this barrier is already significant in two-player games, in this section we shall focus on the two-player case with a finite action set $[K]$. We refer to the row player as agent $A$, column player $B$, and use indices $i,j \in [K]$ to denote their pure actions. 

\subsection{``Diamond in the Rough''  -- A Benchmark Game for Multi-Agent Learning} 
We introduce an interesting class of two-player \emph{dominance solvable} games, which serves as a challenging benchmark for multi-agent learning. For reference convenience, we call it ``\emph{Diamond in the Rough}'', whose meaning should become clear in the following formal definition. 



\begin{definition}[The Diamond-In-the-Rough (DIR) Games]\label{def:diam}
A \dir{} (DIR) game is a two-player game parameterized by $(K, c)$. Each agent have $K$ actions and utility function 
\begin{equation}
    u_1(i,j)=\begin{cases}
    i/\rho & i \leq j+1\\
-c/\rho &  i > j+1
\end{cases},  \qquad 
u_2(i,j)=\begin{cases}
j/\rho &  j \leq i \\
-c/\rho & j > i
\end{cases}.  
\end{equation}
where $c>0$ and $\rho=\max\{K, c \} $ is for normalization purpose. Hence, the payoff matrix of the \gameshort$(K, c)$ game is given by
\begin{equation}\label{eq:lowerboundinstance}
\frac{1}{\max\{K, c \} }\begin{bmatrix} 
(1,1) & (1, -c) & (1, -c) & \cdots & (1, -c) \\
(2,1) & (2,2) & (2, -c) & \cdots  & (2, -c) \\
(-c, 1) & (3,2) & (3,3) & \cdots  & (3, -c) \\
\vdots & \vdots & \ddots & \ddots & \vdots \\
\vdots & \vdots & \vdots &  (K-1, K-1) & (K-1, -c) \\
(-c, 1) & (-c,2) & \cdots & (K, K-1) & (K, K) \\
\end{bmatrix} .   
\end{equation}



\end{definition}
The DIR game exhibits a ``nested'' dominance structure, which makes it challenging to play rationally. Specifically, observe that $A$'s action 2 dominates  action 1. However, this is not the case for $B$ since if $A$ were to play action 1, $B$'s utility $-c$ of action 2 is significantly worse than utility $1$ of action 1. Nevertheless, after $A$ eliminates action 1, $B$'s action 2 starts to dominate $B$'s action 1. This property holds in general for DIR game --- $B$'s action $j+1$ dominates her action $j$ \emph{only when} $A$ eliminates his actions $\{ 1,\cdots, j \}$ and similarly $A$'s action $i+1$ dominates action $i$ \emph{only when} $B$ eliminates her actions $\{ 1,\cdots,i-1 \}$. In the end, the real ``diamond'' is hidden at the action pair $(K,K)$, which is the best for both agents. However, to find this ``diamond'', both agents have to sequentially remove all the ``rough'' actions $1,2,\cdots, K-1$. This is thus the name ``\emph{Diamond in the Rough}''. 

As we will demonstrate both theoretically and empirically, \emph{the DIR game highlights a fundamental challenge in multi-agent learning} --- i.e., whether an agent's action is good or bad depends on what actions her opponents take, and such inter-agent externality makes it challenging to learn the optimal decisions. We end this subsection by summarizing a few useful properties of any \gameshort$(K,c)$ game.

\begin{proposition}\label{prop:diamprop}
The following properties hold for any \gameshort$(K,c)$ game:  
\begin{enumerate}
\item The game is dominance solvable by alternatively eliminating $A$ and $B$'s action in order $1,2,...$, until reaching the last strategy profile $(K,K)$. The elimination length $L_0=2K-2$.~\footnote{While the equilibrium is obvious in \gameshort{}, we remark that we can easily swap the rows and column making it difficult to determine even the elimination path under bandit feedback and strategic learning setting (without communication).}
\item The strategy profile $(K,K)$ is the unique CE (and thus the unique NE as well). Moreover, both agents achieve the maximum possible utility at this equilibrium.   
\end{enumerate}
\end{proposition}
 

\subsection{No-swap Regret $\not \Rightarrow$ Efficient Iterated Dominance Elimination}
 
 Our following theorem shows that for \emph{any} small $\epsilon$, there exist  \gameshort{}($\log(\frac{1}{\epsilon}),c$) games for some constant $c$, in which an $\epsilon$-CE may \emph{never} play the unique CE $(K,K)$ and, moreover, will lead to a much smaller utility than the agent's equilibrium utility.\footnote{The key property here is that the game payoffs depend on $\epsilon$  \emph{logarithmically}.  Such an instance with payoffs that depends on $\epsilon$ linearly would be less surprising and also easier to construct.}  
\begin{theorem}\label{thm:diamwelfare}
For any $\epsilon > 0$ and  any \gameshort$(K,c)$ game satisfying $\log(1/\epsilon) \leq  (2K-2) \log(c) $,  the game always has an $\epsilon$-CE which plays the (unique) CE strategy  $(K,K)$ with probability $0$. Moreover, the welfare of this $\epsilon$-CE is at most $ \frac{1+ \lceil \log(1/\epsilon)/\log (c) \rceil}{2K} $ fraction of the equilibrium welfare. 
\end{theorem}

Specifically, by picking $K = \log(1/\epsilon)$ and $\log(c) = 1$, we obtain a \gameshort{}($\log(\frac{1}{\epsilon}),c$) game which admits an $\epsilon$-EC that will put $0$ probability at the unique CE $(K,K)$ and has utility at most half of the players' equilibrium utilities. This may appear quite counter-intuitive at the first glance, since how come an $\epsilon$-EC be so such ``far away'' from the real and unique  CE. This turns out to be due to a subtle difference --- that is, the $\epsilon$ in ``$\epsilon$-CE'' is measuring the $\epsilon$ difference in \emph{agent utilities},\footnote{ A distribution $\pi \in \Delta_{\cA}$ over action profiles is a $\epsilon-$CE if $\sum_{a_{-n}}\pi(a_n, a_{-n}) \left[ u_n(a'_n, a_{-n})  - u_n(a_n, a_{-n})\right] \leq \epsilon$ for any two actions $ a_n, a'_n \in \cA_n $ and for any player $n$. When $\epsilon=0$, this definitions degenerate to the standard  CE.  } whereas the ``far away'' conclusion reflected in Theorem \ref{thm:diamwelfare} is measuring the true distance between agent's \emph{action probabilities}.  Though in the limit  as $\epsilon\to 0$, the $\epsilon$-CE will tend to an exact CE, Theorem \ref{thm:diamwelfare} suggests that agent's strategies and equilibrium utilities can be far from the exact CE even when $\epsilon>0$ is extremely small compared to the game payoff. Therefore,  the fact that an agent does not have much incentive to deviate   when at an $\epsilon$-CE  does not imply that the action it plays is close to the (exact) CE action profile, neither implies his utility is close to the CE utility.  \emph{This insight also explains why classic no regret learning algorithm designed based on maximizing accumulated rewards may perform poorly for the task for iterated dominance elimination}, which will be formally proved in our next theoretical result as well as    further justified in our numerical results in Section \ref{append:exp}.

To concretize the message in Theorem \ref{thm:diamwelfare}, consider a very small \game{} with $c=K=10$. 
With the state-of-the-art $O(\log^4 T /T)$  swap regret algorithm by ~\citet{anagnostides2021near}, the empirical distribution of the no-regret learning agents is guaranteed to be a $10^{-9}-$CE after $T=10^{13}$ rounds. However, since $\log(1/\epsilon) = \log(10^{9}) < (2K-2)\log(c)$,   Theorem \ref{thm:diamwelfare} implies that  this $10^{-9}-$CE may still never play the equilibrium strategy $(K,K)$ and its welfare is at most $ \frac{1+ \lceil \log(1/\epsilon)/\log (c) \rceil}{2K} = \frac{1}{2}$ of the equilibrium welfare. Notably, $10^{13}$ rounds of sequential interactions could take several days for a modern CPU to simulate. 
Our experimental results in Appendix \ref{append:exp} further confirm  the mathematical analysis in Theorem \ref{thm:diamwelfare} --- our simulation  shows that, surprisingly, the player actions generated by no-swap regret algorithms will get stuck on the first few actions for both players and  remain  far from the unique CE $(K,K)$ even after $10^8$ rounds.    

\subsection{Exponential-Time Convergence of \Algoclass{} Algorithms} 
We now show that a broad and natural class of online learning algorithms, dubbed {\it \algoclass{}} algorithms, fails to eliminate all iteratively dominated actions efficiently. Roughly speaking, a \algoclass{} algorithm has the following property: at each time step $t$, if the accumulated reward from action $i$ exceeds that from action $j$, the learning algorithm will be more likely to play action $i$ than action $j$. We call online learning algorithms with such property ``\algoclass{}'' learning algorithms. Perhaps unsurprisingly,  many celebrated learning algorithms are \algoclass{}, e.g., Follow-the-Perturbed-Leader (FTPL), and the entire class of Dual Averaging (DA) algorithms with symmetric mirror maps including Exponential Weight (EW), lazy gradient descent (LGD) and fictitious play (see our detailed discussion in Appendix \ref{append:nonconverge}). Formally, we define \algoclass{} algorithms as follows.

\begin{definition}\label{def:fair_algo}
Let $\y_t=(y_1(t),\cdots,y_K(t))$ be the vector that stores the (possibly weighted) accumulated rewards for all the actions before round $t$, and $\p_t=(p_1(t),\cdots,p_K(t))$ be the probability distribution of the algorithm taking each action at round $t$. We call an algorithm \algoclass{} if  $y_i(t)>y_j(t)$ implies $p_i(t)\geq p_j(t)$ for any $i,j\in[K], t>0$. 
\end{definition}

In other words,  a \algoclass{} algorithm  maps the accumulated score vector $\y_t$ at each round to a distribution $\p_t=(p_1(t),\cdots,p_K(t))\in\Delta_K$ with some ``order-preserving" function $F$ (whose formal definition can be found in Appendix \ref{append:nonconverge}) and then randomly samples an action from $\p_t$ as the next move. The algorithm is also allowed to specify a learning rate sequence $\{\eta_t\}$ to accumulate the total rewards from each round. For convenience of our arguments, we formulate this general class of online learning algorithms into the following Algorithm \ref{al:fair}. 

\begin{algorithm}[h]
    \caption{The \Algoclass{} Algorithm Framework}
     \label{al:fair}
     \begin{algorithmic}[1]
         \State \textbf{Input:} An order-preserving function $F: \mathcal{Y}\rightarrow \Delta_K$, learning rate sequence $\{\eta_t>0\}$.

        \State $\y_1 \gets (0, \cdots, 0)$\;
        
        \For{$t = 1\dots T$}
            \State Compute $\p_t = F(\y_t)$\;
            
            \State Draw an action $i_t$ from the distribution $\p_t$.\;
            
            \State Receive the expected payoff $\tilde{\u}_t=(\tilde{u}_1(t),\cdots,\tilde{u}_K(t)) $ for each action $i$ from the first-order oracle.\;
            
            \State Update $\y_{t+1} = \y_t + \eta_t \tilde{\u}_t$. \;
            
        \EndFor
    \end{algorithmic}
\end{algorithm}
 
Notably, the notion of ``\algoclass{}'' applies to both the full-information setting (i.e., the rewards of all actions are revealed after taking an action) and the bandit setting (i.e., only the reward of the taken action is revealed). In the following analysis, we show that even with access to  full-information feedback, any \algoclass{} algorithm needs at least exponentially many rounds to eliminate all iteratively dominated actions. 

We consider the situation where all agents are running \algoclass{} algorithms with typically adopted non-increasing learning rates.\footnote{Variants of EW (a.k.a., multiplicative weight updates) have been proved to converge to equilibria in other games such as potential games \citep{cohen2017learning} and concave games \citep{Bravo2018bandit}. } That is, at any round $t$, each agent will do the standard update for any action $i$ using estimated reward $\tilde{u}_{i}(t) $ with learning rate $\eta_t$ that is non-increasing in $t$.
Perhaps surprisingly, even with perfect gradient feedback (as oppose to the noisy gradient from bandit feedback)\footnote{The perfect gradient of each agent $n$ in the game is a vector with each entry, $\tilde{u}_{a_n}(t) = \sum_{a_{-n}} p_{a_{-n}}(t) u_n(a_n, a_{-n}) $, equivalent to the expected payoff of each action $a_n$ given the opponents' strategy profile, since the loss is a linear function of the payoff. The noisy gradient estimated from bandit feedback can be found in our newly designed algorithm in the following section. }, agents following the \algoclass{} algorithms will take provably exponential rounds to converge to the unique NE in DIR games.

\begin{theorem}\label{thm:nonconverge} Consider the \gameshort{}$(K,c)$ game with any $K\geq 3$ and $c \geq 3K^2$. Suppose two agents both follow a \algoclass{} algorithm \ref{al:fair} 
equipped with a non-negative, bounded, \emph{non-increasing} learning rate sequence $\{\eta_t\}_{t=1}^{\infty}$. Then agent $B$ will place at most $1/2$ probability on the (unique) pure NE strategy at any round  $ t \leq 3^{K-2}$.\footnote{This rules out even the (weaker) average-sequence convergence (in distribution sense) to the unique NE of the game.} 

\end{theorem} 
Theorem \ref{thm:nonconverge} proves the inefficiency of the family of \algoclass{} algorithms with a non-increasing learning rate sequence in terms of eliminating iteratively dominated strategies. We remark that the requirement $c\geq 3K^2$ is only necessary to the proof technique and does not imply the \gameshort{} game is easy to solve when $c< 3K^2$. As we will demonstrate later in the experiments, the choice of $c=O(K)$ already results in an extremely slow empirical convergence rate for Exp3 algorithm and its variants. 
\paragraph{Additional Discussions on Connection to Related Works.} 
\citet{cohen2017hedging} showed that under full information feedback, the EW algorithm with learning rate $\eta_t=\frac{1}{t^b}$ for $ 0<b<1$ eliminates dominated actions exponentially fast. They remarked that their proof can be extended by induction to argue
that the sequence of play induced by EW will ultimately eliminate all iteratively dominated strategies of the game. Our Theorem \ref{thm:nonconverge} shows that the number of rounds needed by their algorithm will, unfortunately, be exponential in the elimination length. That is, it is easy to eliminate dominated actions for one iteration, but difficult to iteratively eliminate them all. \citet{cohen2017hedging,mertikopoulos2019learning}    showed   the EW algorithm with proper learning rate converges to a pure NE exponentially fast with high probability if   that equilibrium satisfies a global variational stability.   Our negative result does not contradict their results because  the global variational stability condition does not hold in \gameshort{} games. The verification of this claim can be found in Appendix \ref{append:variational}. \cite{laraki2013higher} considered the replicator dynamics, the continuous version of EW, and showed the probability of playing any iteratively dominated action shrinks to $0$ over time. This result matches ours that all iteratively dominated actions will be removed as $t\to \infty$. However, the convergence rates in continuous-time dynamics are not necessarily meaningful, as $t$ can be reparametrized arbitrarily; it is unclear how a time-dependent rate can be translated to its corresponding rate in a discrete-time framework as a function of iterations, because infinitely many iterations are required to simulate a continuous dynamics. Therefore, our exponential time convergence result is a more explicit characterization of the learning barrier in iterated dominance elimination.


\subsection{Proof of Theorem \ref{thm:nonconverge}}\label{subsec:proof-of-exp-converge}



\paragraph{Overview of the Proof.} The proof of Theorem \ref{thm:nonconverge} is involved and requires new ideas since we are not aware of any result of similar spirit in the literature. We start by highlighting the key ideas here before diving into the technical arguments. Without loss of generality, we  consider player B and let the action set be $\{b_1,\cdots,b_K\}$. The key to our proof is to show the existence of a constant $c_0>1$ such that if B's action probabilities $p_B(t)$ at time $t$ concentrate on $\{b_1, \cdots, b_n\}$ for the first $T_{n}$ rounds, then it must also concentrate on $\{b_1, \cdots, b_{n+1}\}$ for the first $T_{n+1}=c_0T_{n}$ rounds. The intuition behind this fact comes from the property of the DIR game: any pure action looks very profitable before it becomes iteratively dominated and thus could have accumulated a very large score before any \algoclass{} algorithm starts to penalize it. As a result, any \algoclass{} algorithm has to suffer additional rounds \emph{proportional to the current time step} to eliminate the next action, which incurs an exponential convergence time. A intricate induction is needed to make this intuition a formal argument, which we formally show next. 

\noindent
\textbf{Step 1: Preparations for the Proof. } 
We consider the DIR game with $c\geq 3K^2>K$, so the normalization factor in Equation \eqref{eq:lowerboundinstance} is $1/c$. Both agents follow Algorithm \ref{al:fair} with a non-increasing learning rate sequence $\{\eta_t\}_{t=1}^{\infty}$ in a repeated DIR game.  Let $u_{A,i}(t), u_{B,i}(t), p_{A,i}(t), p_{B,i}(t)$ denote agent $A$ and agent $B$'s payoffs\footnote{For simplification of notation, we omit the bar `` $\tilde{}$ '' above $u$, but it should still be interpreted as the expected payoff of an action given opponent's strategy profile.} and probabilities of picking the $i$-th action at round $t$, respectively. We further let  $$y_{A,i}(t)=\sum_{s=1}^{t-1} u_{A,i}(s)\eta_s, \qquad y_{B,i}(t)=\sum_{s=1}^{t-1} u_{B,i}(s)\eta_s$$  be the $i$th arm's accumulated weights of agent $A$ and agent $B$ till round $t$. Then, according to Algorithm \ref{al:fair},  
$$(p_{A, 1}(t),\cdots,p_{A, K}(t))=F(y_{A, 1}(t),\cdots,y_{A, K}(t)),$$ $$(p_{B, 1}(t),\cdots,p_{B, K}(t))=F(y_{B, 1}(t),\cdots,y_{B, K}(t)).$$

\noindent \textbf{Step 2: The Dueling Lemmas and Their Proofs. } 
 At a high level, by leveraging the ``order-preserving'' property from the definition of \algoclass{} algorithms, our proof employs a complex induction argument that if both players have significant probabilities to play the first $n$ actions within time $T_n$, then they must also have significant amount of probabilities to play the first $n+1$ actions within time $(1+\frac{2c}{3K^2})T_n$. This intuition is formalized by the following two lemmas. The first lemma says that if $B$ has constant probability to play actions from $1,2,\cdots, n$ within the first $(1+\frac{2c}{3K^2})^{n-1}$ rounds, then any dual averaging algorithm used by agent $A$ must play actions $1,2,\cdots, n, n+1$ with some constant probability during the same period. The second lemma is an analogous (though subtlely different) conclusion that uses $A$'s behavior to argue $B$'s trajectory properties. It is such dueling between $A$'s and $B$'s behaviors --- thus the name of \emph{dueling lemmas} --- that make their convergence to rationalizable actions exponentially slow.  

We believe these dueling lemmas reveal the intrinsic difficulty of using \algoclass{} algorithms for iterated dominance elimination, and thus provide a formal proofs for them below. The key challenge in their proofs is to manage the inter-agent utility externalities by setting the right parameter scales at which the two agents' strategies duel. 


\begin{lemma}\label{lm:inductive_argument1}
For any $1\leq n\leq K-2$, if there exists $T_n\geq 1$ such that $
    \sum_{j=1}^n p_{B,j}(t) \geq \frac{1}{K-n+1}$  for any $  t\leq T_n$, 
then we must have $
        \sum_{i=1}^{n+1} p_{A,i}(t)\geq \frac{1}{K-n}, \forall t\leq T_{n+1}$ for  $T_{n+1}=(1+\frac{2c}{3K^2})T_n$.   
\end{lemma}
\begin{lemma}\label{lm:inductive_argument2}
For any $1\leq n\leq K-2$, if there exists $T_{n+1} \geq 1$ such that $ \sum_{i=1}^{n+1} p_{A,i}(t)\geq \frac{1}{K-n}$ for any $t\leq T_{n+1},$ then we must   have  $ 
        \sum_{j=1}^{n+1} p_{B,j}(t)\geq \frac{1}{K-n}, \forall t\leq T_{n+1}.$
\end{lemma}
\begin{proof}\textbf{of Lemma \ref{lm:inductive_argument1}}.  
At each step $t\leq T_n$, using the assumed condition that $\sum_{j=1}^n p_{B,j}(t) \geq \frac{1}{K-n+1}$ and $\frac{K}{c} \leq \frac{1}{3K}$,  we have for any  $ i=n+2,\cdots,K$ 
    \begin{align*}
        u_{A,i}(t)
        & =    \sum_{j=1}^n p_{B,j}(t)\cdot(-1)+(1-\sum_{j=1}^n p_{B,j}(t))\cdot(\frac{i}{c}) & \mbox{by Def. of DIR games}  \\ 
        &\leq   \sum_{j=1}^n p_{B,j}(t)\cdot(-1)+(1-\sum_{j=1}^n p_{B,j}(t))\cdot(\frac{K}{c}) &   \\  
        & \leq   \frac{1}{K-n+1}\cdot(-1)+(1-\frac{1}{K-n+1})\cdot(\frac{K}{c})  & \mbox{by assumed conditions} \\ 
        & \leq    -\frac{1}{K-n+1}+ \frac{K-n}{K-n+1}\cdot \frac{1}{3K} & \mbox{since $c \geq 3K^2$} \\ 
        &  <   \frac{-2}{3(K-n+1)} & \mbox{since $\frac{K-n}{K}<  1$}  \\
        & \leq -\frac{2}{3K}. & \mbox{$\frac{-2}{3(K-n+1)}$ decreases in $n$}
    \end{align*} 

Suppose player $A$ uses any DA algorithm with learning rate $\{ \eta_t \}_{t=1}^{\infty}$. Using the above strict upper bound for $u_{A,i}(t) (<-\frac{2}{3K})$, the cumulative rewards of any action $i \in \{n+2,\cdots, K\}$ at time step $t=T_n+1$ can  be upper bounded as 
    \begin{align}\label{eq:y-Ai-upperbound}
        y_{A,i}(T_n+1)&= \sum_{t=1}^{T_n}\eta_t u_{A,i}(t) < -\frac{2\sum_{t=1}^{T_n} \eta_t}{3K}, \quad \forall i=n+2,\cdots,K.
    \end{align}
Now let $T_{n+1}=(1+\frac{2c}{3K^2})\cdot T_n$ and consider any $t\leq T_{n+1}$. First, if $t \leq T_n+1$, we have $ y_{A,i}(t)= \sum_{\tau=1}^{t-1}\eta_\tau u_{A,i}(\tau) < -\frac{2\sum_{\tau=1}^{t-1} \eta_\tau}{3K} \leq 0$. We now consider $ T_n + 1 <  t \leq T_{n+1}$. For any $i=n+2,\cdots,K$,  we have 
    \begin{align*}  
        y_{A,i}(t) &= y_{A,i}(T_n+1) + \sum_{s=T_n+1}^{t-1}\eta_t u_{A,i}(s)   \\ 
        &<-\frac{2\sum_{t=1}^{T_n} \eta_t}{3K}+ \sum_{s=T_n+1}^{t-1}  \eta_{T_n+1}\cdot \frac{K}{c}  & \mbox{by Ineq. \eqref{eq:y-Ai-upperbound} and $\eta_t \leq \eta_{T_{n+1}}$} \\  
        &\leq -\frac{2\sum_{t=1}^{T_n} \eta_t}{3K}+\frac{K}{c}\cdot \eta_{T_n+1}\cdot (\frac{2c}{3K^2}  T_n)  & \mbox{since $t-T_n \leq \frac{2c}{3K^2}  T_n$} \\  
        &= -\frac{2\sum_{t=1}^{T_n} (\eta_t-\eta_{T_n+1})}{3K}\leq 0,  & \mbox{due to non-increasing learning rate}
    \end{align*}
    
Therefore, $y_{A,i}(t) < 0$ for any $ t \leq T_{n+1}$.  However, note that $y_{A,1}(t)\geq 0$ for any $ t>0$, we thus conclude that $y_{A,i}(t)<y_{A,1}(t)$ for any action $i\geq n+2$ at any time $t\leq T_{n+1}$. 


Apply the definition of \algoclass{} algorithms, we obtain $p_{A,i}(t)\leq p_{A,1}(t)$ and as a result,
    \begin{align} \notag 
          K-n-1  
        &\geq (K-n-1) \cdot [p_{A,1}(t) + \sum_{i=n+2}^{K} p_{A,i}(t)] & \\  \notag 
          & \geq \sum_{i=n+2}^{K} p_{A,i}(t) + (K-n-1)\cdot [\sum_{i=n+2}^{K} p_{A,i}(t)] & \mbox{by $p_{A,1}(t) \geq p_{A,i}(t), \forall i \geq n+2$}\\ \label{eq:duelinglemma-bound}
          & = (K-n)[\sum_{i=n+2}^{K} p_{A,i}(t)].  
    \end{align}
This implies $\sum_{i=n+2}^{K} p_{A,i}(t)\leq \frac{K-n-1}{K-n}, \forall t\leq T_{n+1}$ and thus $ \sum_{i=1}^{n+1} p_{A,i}(t)\geq \frac{1}{K-n},$ as desired.  
 \end{proof}
 
 \begin{proof}\textbf{of Lemma \ref{lm:inductive_argument2}}.  
Using the assumed condition $ \sum_{i=1}^{n+1} p_{A,i}(t)\geq \frac{1}{K-n}$ for any $ t\leq T_{n+1}$, we can   derive an upper bound for $u_{B,j}(t)$ for any $j = n+2, \cdots, K $, as follows     
    \begin{align*}
        u_{B,j}(t) &= \sum_{i=1}^{n+1} p_{A,i}(t)\cdot(-1)+(1-\sum_{i=1}^{n+1} p_{A,i}(t))\cdot(\frac{j}{c}) & \mbox{by Def. of DIR games}\\  
        &\leq\sum_{i=1}^{n+1} p_{A,i}(t)\cdot(-1)+(1-\sum_{i=1}^{n+1} p_{A,i}(t))\cdot(\frac{K}{c}) \\ 
        & \leq \frac{1}{K-n}\cdot(-1)+(1-\frac{1}{K-n})\cdot(\frac{K}{c}) \\ 
        & \leq -\frac{1}{K-n}+ \frac{K-n-1}{K-n}\cdot \frac{1}{3K} \\ 
        & < -\frac{2}{3(K-n)} \leq -\frac{2}{3K}.
    \end{align*} 
    
 Hence, we similarly conclude that $y_{B,j}(t)<0\leq y_{B,1}(t)$ for any $j\geq n+2$ and $t\leq T_{n+1}$, which yields $p_{B,j}(t)\leq p_{B,1}(t)$. Then we similarly follow Inequality \eqref{eq:duelinglemma-bound} to obtain $\sum_{j=n+2}^{K} p_{B,j}(t)\leq \frac{K-n-1}{K-n}$ and thus $
        \sum_{j=1}^{n+1} p_{B,j}(t)\geq \frac{1}{K-n}, \forall t\leq T_{n+1}$,    which completes the proof. 
\end{proof}

\noindent \textbf{Step 3: Concluding the Proof. }  
Armed with the dueling lemmas   \ref{lm:inductive_argument1} and   \ref{lm:inductive_argument2}, we are now ready to conclude the proof. Specifically, the following is a direct corollary of the two lemmas. 
\begin{corollary}\label{cor:dueling}
For any $1\leq n\leq K-2$, if there exists $T_n\geq 1$ such that $
    \sum_{j=1}^n p_{B,j}(t) \geq \frac{1}{K-n+1}$  for any $  t\leq T_n$, 
then   for  $T_{n+1}=(1+\frac{2c}{3K^2})T_n$ we must   have  $ 
        \sum_{j=1}^{n+1} p_{B,j}(t)\geq \frac{1}{K-n}$ for any $t\leq T_{n+1}.$ 
\end{corollary} 
Note that the ``if'' condition in the above corollary clearly holds for $n=1$, in which case we can verify that $T_1=1$ satisfies  $\sum_{j=1}^n p_{B,j}(t) =p_{B,1}(1) \geq \frac{1}{K} = \frac{1}{K-n+1}$  for any $  t\leq T_n)$  since both players start by taking actions uniformly at random. Applying corollary \ref{cor:dueling} inductively until any $n\leq K-2$, we  obtain that  $\sum_{j=1}^{n+1} p_{B,j}(t) \geq \frac{1}{K-n}, \forall t\leq T_{n+1} = (1+\frac{2c}{3K^2})^{n}T_{1}$. Plugging in $T_1 = 1$ and the upper bound $K-2$ of allowed value for $n$, we thus have 

 $$\sum_{j=1}^{K-1} p_{B,j}(t) \geq \frac{1}{2}, \forall t\leq T_{K-1} = (1+\frac{2c}{3K^2})^{K-2}$$
Since $c \geq 3K^2$, so $T_{K-1} \geq 3^{K-2}$. The above inequality implies that within the first $3^{K-2}$ time steps, player $B$ will never  play the unique equilibrium action $K$ with probability larger than $1/2$. This rules out the convergence of the two agents' strategies --- either in the average sense or last-iterate sense ---  to the NE within  $3^{K-2}$ time steps. 
Note that since this statement holds deterministically  since we are in the simpler setup in which the agents have full feedback, and thus can fully observe the (expected) payoff (i.e., perfect gradient feedback).  Since $(K, K)$ is the unique NE of this game,  Moreover, a larger value of $c$ will imply a larger $T_{K-1} = (1+\frac{2c}{3K^2})^{K-2}$ value and would further slow down the convergence rate.


\section{\model{} and its Efficiency towards Rationalizability} \label{sec:converge}

\subsection{The  \fullmodel{}  (\model{}) Algorithm}
Motivated by the barriers in Section \ref{sec:barr},  we now introduce a novel algorithm\emph{ \fullmodel{} } (\model{}) described in Algorithm \ref{al:hdg},    which turns out to provably guarantee efficient  elimination of all iteratively dominated actions, with high probability. 

\begin{algorithm}[t]
    \caption{\texttt{Exp3} with \texttt{D}iminishing \texttt{H}istory (\model{})}
     \label{al:hdg}
    \begin{algorithmic}[1] 
        \State \textbf{Input:} Number of actions $K$, parameter $\epsilon_t$, $\beta$ .
        \State \textbf{Initialization:} $y_i(0)=0, \forall i\in [K]$.
            \For{$t = 0,1, \dots, T$} 
            \State Set 
            $$
            p_i(t) = (1-\epsilon_t) \frac{ \exp (y_i(t)) }{ \sum_{j\in[K]} \exp(y_j(t)) } + \frac{\epsilon_t}{K}  \quad \forall i\in [K].
            $$
            \State Draw action $i_t$ randomly accordingly to the probability distribution $(p_1(t), p_2(t), \dots p_K(t))$.
            \State Observe realized payoff $u_{i_t}(t)$.
            \State Compute \emph{unbiased} payoff estimator $\tilde{u}_{i}(t) = \frac{u_{i}(t)}{p_{i}(t)}\cdot \mathbb{I}(i = i_t), \forall i \in [K]$. 
            \State Update 
            $$
            y_i(t+1) =  
              (\frac{t}{t+1})^{\beta} \cdot  y_i(t) + \tilde{u}_{i}(t), \quad  \forall i\in [K].
            $$
            \EndFor
    \end{algorithmic}
\end{algorithm}

\model{} is an EW-style algorithm but with the following important characteristics:
\begin{enumerate}[leftmargin=*]
    \item \model{} uses an unbiased payoff estimator $\tilde{u}_{i_t}(t) = \frac{u_{i_t}(t)}{p_{i_t}(t)}$. This turns out to be crucial  since  our proof  has to upper bound the absolute value $\left|\sum_{t=1}^T \gamma_t(\tilde{u}_{a}(t)-u_{a}(t))\right|$ whereas standard single-agent bandit problems  only need a one-side upper bound for $\sum_{t=1}^T \gamma_t(\tilde{u}_{a}(t)-u_{a}(t))$ and thus a biased conservative estimation of $\tilde{u}_{i_t}$ there can be helpful (actually, is essential for some algorithms). However, similar conservative reward estimators  appear  difficult to work in our problem because it may let a dominating action lose its advantage. This highlights an interesting difference between eliminating iteratively dominated actions in multi-agent setup and learning optimal actions in single-agent setup.
    \item During reward updates at Step 8, \model{} will always discount \emph{each action}'s previous reward $y_i(t)$ by a factor $(\frac{t-1}{t})^{\beta}$ for some carefully chosen parameter $\beta$, \emph{regardless of whether this action is taken at this round or not}. Therefore, the rate of historical rewards will gradually diminish, which thus leads to the name of our algorithm. 
\end{enumerate}     

\paragraph{Effective Learning Rates.} Note that the update in Step 8 of Algorithm \ref{al:hdg} only captures the recursive relation between $y_i(t+1)$ and $y_i(t)$. From this recursion, we can easily derive how $y_i(t)$ depends on all previous payoff estimation $\tilde{u}_i(\tau)$ for $\tau=0,1,\cdots, t$, which is the follows, 
\begin{equation}\label{eq:sum_loss}
y_i(t+1)=   \sum_{\tau=0}^{t} (\frac{\tau}{t})^{\beta} \tilde{u}_i(\tau).
\end{equation} 
For notational convenience, we call  $\gamma_{\tau}^{(t)}=(\tau/t)^{\beta}$ the \emph{effective} learning rate. Notably, the learning rate $\gamma_{\tau}^{(t)}$ for any fixed past payoff estimation $\tilde{u}_i(\tau)$ \emph{dynamically decreases} as the round $t$ becomes large. This means that the weight of the historical estimation $\tilde{u}_i(\tau)$ will become smaller and smaller as time goes. In other words, the algorithm exhibits \emph{recency bias} and gradually ``forgets'' histories and always relies more on recent payoff estimations.  


We end this section by comparing \model{} with previous Exp3-style algorithms. 
 In standard Exp3 algorithm \citep{auer2002nonstochastic}, the \emph{effective} learning rate is exactly its learning rate $\gamma_t$, which is set to a constant $\sqrt{\frac{\log K}{KT}}$ or $O(\sqrt{\frac{\log K}{Kt}})$ to guarantee a sub-linear regret. Some other variants\footnote{See~\cite{neu2015explore,mertikopoulos2019learning, cohen2017learning, Bravo2018bandit}} of Exp3 use non-increasing effective learning rate $\gamma_t$. 
 However,  \model{} differs from these algorithms in at least two key aspects: (1) its effective learning rate $\gamma_\tau^{(t)}$ is increasing in $\tau$, i.e., biased towards recent reward estimations; (2) $\gamma_\tau^{(t)}$ will be re-scaled every time  $t$ increases, i.e., a new observation comes. The first deviation is justified by our Theorem \ref{thm:nonconverge} since DA with any decreasing learning rate necessarily suffer exponential convergence.  We note that increasing learning rate has  been recently studied  for single-agent bandit problems \citep{bubeck2017kernel,lee2020bias,agarwal2017corralling} and for faster convergence to coarse correlated equilibrium in games \citep{syrgkanis2015fast}.\footnote{ \citet{syrgkanis2015fast} use optimistic follow the regularized leader (OFTRL) with recency bias. The obtained algorithm with entropy regularizer can be viewed as EW with dynamically increasing learning rate. }  Unfortunately, our experimental results in Appendix \ref{append:exp} show that these algorithms fail to efficiently eliminate all iteratively dominated actions, which illustrates that careful design of learning rate is necessary for efficient iterative dominance elimination. 
 
 
 

\subsection{Efficient Convergence of \model{} under Noisy Bandit Feedback}
We now present the theoretical guarantee for \model{}. Due to randomness, no algorithm guarantees dominance elimination for certain. Thus, we introduce a natural notion of \emph{essential elimination}: 
 \begin{definition}[$\varepsilon$-\emph{Essential Elimination}]\label{df:essential_elim}
We say that an action $i\in\cA_n$ is $\varepsilon$-\emph{essentially eliminated} at time step $T$ if agent-$n$'s probability of playing $i$ satisfies
$p_i(T) \leq \frac{\varepsilon}{4KN}$.
\end{definition}
Note that if all the actions in $E_{L_0}$ are \emph{essentially eliminated} at time step $T$, the mixed strategy $x(T)$ given by the probability distribution $p(T)$ of all agents satisfies $\|x(T)-x^*\|_1\leq \frac{\varepsilon}{4KN}\cdot 2(KN-N)<\frac{\varepsilon}{2}.$ Therefore, $\varepsilon$-\emph{Essential Elimination} implies last-iterate convergence. We are now ready to give the convergence guarantee for \model{} in Theorem \ref{thm:main_convergence}:
 \begin{theorem}\label{thm:main_convergence}
Consider any game $\cG(\cN, \cA, u)$ with elimination length $L_0$ and elimination set $\{ E_l \}_{l=1}^{L_0}$. Suppose
all agents in $\cN $  run \model{}  with parameters $\beta>0$ and  $\epsilon_t=t^{-b}$ for some $ b \in (0,1)$. Then for any $\varepsilon,\delta\in (0,\frac{1}{2})$, any action in $E_l\setminus E_{l-1}$ will be $\varepsilon$-\emph{essentially eliminated} with probability at least $1-2|E_l|(T_l+s)^2\delta, \forall s\geq 0$, from time step $t=T_l$ to $t=T_l+s$,  where the sequence $\{T_l\}_{l=1}^{L_0}$ is defined recursively below: 
\begin{enumerate}
    \item $T_1$ is an integer such that for any $t\geq T_1$, 
    \begin{equation}\label{eq:suff_T1}
\frac{t^{-b}}{K} +  \exp\left(4\left(\sqrt{\frac{eK(1+\sigma^2)}{1+2\beta+b}}\right)\log^{\frac{1}{2}}\frac{2K}{\delta}\cdot t^{\frac{1+b}{2}} -\frac{\Delta t}{16(1+\beta)} \right)  < \frac{\min\{\varepsilon,\Delta/2\}}{4KN}.    
\end{equation}
    \item For any $l\geq1$, 
    \begin{equation}\label{eq:eq4Tk}
    T_{l+1}\geq\max\Bigg\{ (1+\frac{8}{\Delta})^{\frac{1}{1+\beta}}\cdot T_l, T_l+ \frac{(1+\beta)^2(4+\Delta)^2(8+\Delta)^2 }{4(1+2\beta)\Delta^2}\log \frac{1}{\delta}\Bigg\}.
    \end{equation}
\end{enumerate}

Specifically, if all agents run \model{} on $\cG$, then all iteratively dominated actions will be $\epsilon$-essentially eliminated after $T_{L_0}$ number of rounds in the last iterate convergence sense with probability at least $1-2KNT_{L_0}^2\delta$. 
\end{theorem}
To interpret  Equation \eqref{eq:suff_T1} and \eqref{eq:eq4Tk} in Theorem \ref{thm:main_convergence}, it will be  easier if we focus only on its leading terms. It is helpful (though not necessary) to think of $\beta \approx L_0$ and $b \approx 1/2$. When $t$ is large, the exponential term in Equation \eqref{eq:suff_T1} will be  dominated by $t^{-b}/K$ term since its exponent decreases exponentially in $t$. Therefore, we would expect $T_1$ to have  order around $( N \min\{\varepsilon,\Delta\}^{-1})^{1/b}$. In Equation \eqref{eq:eq4Tk}, the second term in the ``$\max$'' is usually dominated by the first term,  therefore    $T_{L_0}$ roughly has order $\Delta^{-L_0/\beta} T_1$. Therefore, Theorem \ref{thm:main_convergence} indicates: 1. \model{} needs polynomially many rounds to $\varepsilon$-essentially eliminate the first set of dominated actions in $E_1$; 2. once a set of iteratively dominated actions are $\varepsilon$-essentially eliminated, \model{} takes polynomially additional steps to eliminate the next set of iteratively dominated actions. Therefore, \model{} guarantees to $\varepsilon$-essentially eliminate all iteratively dominated actions in polynomially many rounds. 
Specifically, the following corollary gives an explicit and formal upper bound for $T_{L_0}$. Its proof is deferred to Appendix \ref{append:converge}.

\begin{corollary}\label{coro:maintheorem}
If all agents run \model{} on game $\cG(\cN, \cA, u)$ with parameters $\beta>0$ and $ \epsilon_t=t^{-\frac{1}{3}}$, then after $\tilde{O}(\max\{N^3,K^{1.5}\}\max\{\varepsilon^{-3},\Delta^{-3}\}(1+\sigma^3)\beta^{1.5}\log^{1.5}\frac{1}{\delta})$ number of rounds\footnote{By convention, $\tilde{O}$ notation omits logarithmic terms. }, all the iteratively dominated actions in $\cG$ will be $\varepsilon$-\emph{essentially eliminated} with probability at least $1-\delta$. 
\end{corollary}
{\bf Parameter Choices. } The optimal choice of $\beta$ depends on  $\Delta$ and $L_0$, which reveal the intrinsic difficulty of iterative dominance elimination in $\cG$. However, if $\Delta, L_0$ is unknown, a conservative choice of $\beta$ is the total number of agent actions $KN$. This always guarantees polynomial round convergence, concretely, in $\tilde{O}(\max\{N^3,K^{1.5}\}\max\{\varepsilon^{-3},\Delta^{-3}\}\beta^{1.5}\log^{1.5}\frac{1}{\delta})$ rounds. We also note that Corollary \ref{coro:maintheorem} instantiated a choice of $b=1/3$. This is to balance the dependence  on $(K,N)$ and $(\varepsilon,\Delta)$. If $b=1/2$, the upper bound would be $\tilde{O}(\max\{N^2,K^{2}\}\max\{\varepsilon^{-2},\Delta^{-4}\}\beta^{2}\log^{2}\frac{1}{\delta})$. When $2=N\ll K$ and $\Delta\sim O(K^{-1})$ in DIR games, a smaller $b$ is preferred.  This will be demonstrated  in our experiments. More generally, a good choice of $b$ will depend on the game structures.     

Corollary \ref{coro:maintheorem} also captures the convergence rate of $\epsilon$ as a function of time horizon $T$, given $b=1/3$. Specifically, $\epsilon = O(1/ \sqrt[3]{T})$ with the caveat that omitted constants in the big $O$ depends on the game parameters $N,K, \Delta,\delta$. If we choose $b=1/2$, the convergence rate will be $\epsilon = O(1/ \sqrt{T})$, however its dependence on the game parameters will be worse.

\subsection{Proof of Theorem \ref{thm:main_convergence}}

\paragraph{Overview of the Proof}

The proof of Theorem \ref{thm:main_convergence}  employs induction to prove the convergence of uncoupled learning. We are not aware of similar ideas in previous multi-agent learning setups, and thus present its formal proof in the next section in case it is of independent interest.  
Our proof deviates from standard single-agent online learning analysis in at least two key aspects.  First, since the effective learning rates are dynamically changing in \model{}, the concentration of the estimated $y_i(t)$ needs to be re-evaluated at each round. To do so, we need to construct a sub-martingale for each round $t$ and argues its high-probability concentration. The second major distinction is that our proof needs to bound the \emph{difference} of the estimated reward for any iteratively \emph{dominating} action $i$ and \emph{dominated} action $j$, whereas standard (single-agent) bandit learning only needs to upper bound the reward estimation for each action $i$.  In the latter case, conservative reward estimation (e.g., in Exp3.P \cite{auer2002nonstochastic}) is helpful. However, a conservative reward estimation may make $i$ lose its dominance advantage in our problem.  Fortunately, due to the unbiased reward estimation in \model{}, we can use Azuma's inequality to bound the concentration from both sides. 

Our proof employs an induction argument. In the base case, we demonstrate that \model{} requires polynomially many rounds to $\varepsilon$-essentially eliminate the first set of dominated actions $E_1$. Subsequently, in the induction phases, we provide further evidence that \model{} takes polynomially additional steps, with high probability, to eliminate the next set of iteratively dominated actions. This leads to a polynomial time last-iteration convergence. The key insight driving each induction phase is that the accumulated scores of the remaining iteratively dominated actions can always be upper-bounded by the number of steps taken so far (which constitutes a polynomial term). As a result, correcting such misconceptions by choosing an appropriate decaying rate to discount the history also takes no more than polynomial time. We present the detailed formal argument next.

\textbf{Step 1: Preparations for the Proof. } Since we will be working mostly on agent's actions, we use notation $a$ and $x$ to denote an agent's pure and mixed strategy in this proof. Our proof relies on the following two technical lemmas, the proofs of which are deferred to Appendix \ref{append:converge} so that they will not distract the core arguments. The first lemma (Lemma \ref{le:decompE}) is standard and bounds the difference between the estimated $\tilde{u}_{a}(t)$ and the realized true $u_{a}(t))$ via the Azuma's inequality. The second lemma  (Lemma \ref{coro:gapy}) is more involved and shows that the accumulated weighted rewards for any dominated action $a$ will be far less than the expected accumulated weighted rewards of any mixed strategy $x \in \Delta_{\cA_n} $ that dominates $a$. 
 \begin{lemma}\label{le:decompE}
For any $a\in\cA, T>0$ and any sequence $\{\gamma_t>0\}_{t=1}^T$, with probability $1-\delta$, we have
    \begin{equation}\label{eq:lmdecompE}
    |\sum_{t=1}^T \gamma_t^{(T)}(\tilde{u}_{a}(t)-u_{a}(t))| < 2\left(\sqrt{K(1+\sigma^2)\log\frac{2}{\delta}\cdot \sum_{t=1}^T \frac{(\gamma_t^{(T)})^2}{\epsilon_t}}\right).
\end{equation}
\end{lemma}

 \begin{lemma}\label{coro:gapy}
For any fixed $T>0$, if an action $a \in \cA_n$ is strictly dominated by a mixed strategy $x=(x_1,\cdots,x_{K}) \in \Delta_{\cA_n}$, then 
\begin{align}\label{eq:deltay} 
y_{x}(T+1)-y_{a}(T+1) \geq \sum_{t=1}^T \gamma_t^{(T)} [u_{x}(t)-u_{a}(t)] - 4\left(\sqrt{\log\frac{2K}{\delta}}\cdot\sqrt{K(1+\sigma^2)\sum_{t=1}^T \frac{(\gamma_t^{(T)})^2}{\epsilon_t}}\right)
\end{align}
holds with probability $1-\delta$. In particular, if $\gamma_t = (t/T)^{\beta}, \epsilon_t = t^{-b}$ and $T>\beta$, we have
\begin{align}\label{eq:deltay2} 
y_{x}(T+1)-y_{a}(T+1) \geq \frac{\Delta \cdot T}{1+\beta} - 4\left(\sqrt{\log\frac{2K}{\delta}}\cdot\sqrt{\frac{eK(1+\sigma^2)}{1+2\beta+b}\cdot T^{1+b}}\right).
\end{align}
\end{lemma}

Armed with the above two lemmas, the proof of Theorem \ref{thm:main_convergence}  follows a carefully tailored induction argument.  


 \noindent {\bf Step 2: Proof of the Base Case. }  First, we prove the elimination of all actions in $E_1$ in the first iteration.  For any dominated action pair $x \succ a$ (i.e., $x$ dominates $a$), we conclude from Lemma \ref{coro:gapy} that with probability $1-\delta$, the probability of playing $a$ satisfies
\begin{align} \notag
p_{a}(t) 
& = \frac{\epsilon_t}{K} +\frac{\exp(y_{a}(t))}{\sum_{a\in \cA_n} \exp(y_{a'}(t))}(1-\epsilon_t)\\ \notag
& \leq \frac{t^{-b}}{K} +\frac{\exp(y_{a}(t))}{\sum_{a'\in\cA_n}x_{a'} \exp( y_{a'}(t))} & \mbox{since $x_{a'} \in [0,1]$ } \\  \notag
& \leq \frac{t^{-b}}{K} +\frac{\exp(y_{a}(t))}{\exp(\sum_{a'\in\cA_n}x_{a'} y_{a'}(t))}  & \mbox{by convexity of $exp(\cdot)$ }\\  \notag
& =\frac{t^{-b}}{K} +\frac{\exp(y_{a}(t))}{\exp(y_{x}(t))} \\ \notag 
& < \frac{t^{-b}}{K} +  \exp\left( 4 \sqrt{\frac{eK(1+\sigma^2)\log\frac{2K}{\delta}}{1+2\beta+b}}\cdot t^{\frac{1+b}{2}} -\frac{\Delta t}{1+\beta} \right) &\mbox{by Eq \eqref{eq:deltay2}} \\  \label{eq:302}
& < \frac{t^{-b}}{K} +  \exp\left(4\left(\sqrt{\frac{eK(1+\sigma^2)}{1+2\beta+b}}\right)\log^{\frac{1}{2}}\frac{2K}{\delta}\cdot t^{\frac{1+b}{2}} -\frac{\Delta t}{16(1+\beta)} \right) \\ \label{eq:133}
& < \frac{\min\{\varepsilon, \Delta/2\}}{4KN},
\end{align}
where 
the last Inequality \eqref{eq:133} holds for any $t\geq T_1$ by the definition of $T_1$. Therefore, from the union bound we conclude that with probability at least $1-|E_1|(T_1+s+1)\delta>1-2|E_1|(T_1+s)^2\delta$, actions in $E_1$ will be $\varepsilon'$-essentially eliminated at each round in $\{T_1, \cdots, T_1+s\}$, where $\varepsilon'=\min\{\varepsilon,\Delta/2\}$.

\noindent  {\bf Step 3: Proof of the Induction Step. }  We now move to the more involved induction step. Assume there exists $T_k$ such that for any $s>0$, actions in $E_k$ are $\varepsilon'$-essentially eliminated during $t\in(T_k, T_k+s]$ with probability $1-2|E_k|(T_k+s)^2\delta$, i.e.,
$$\mathbb{P}\Big[p_a(t) \leq \frac{\varepsilon'}{4KN},\forall a\in E_k, \forall T_k<t\leq T_k+s\Big]\geq 1-2|E_k|(T_k+s)^2\delta,\forall s>0.$$
We investigate how many additional iterations we need to $\varepsilon'$-essentially eliminate actions in $E_{k+1}$. In particular, we will show that if $T_{k+1}$ satisfies Eq \eqref{eq:eq4Tk}, then 
$$\mathbb{P}\Big[p_a(t) \leq \frac{\varepsilon'}{4KN},\forall a\in E_{k+1}, \forall T_{k+1}<t\leq T_{k+1}+s\Big]\geq 1-2|E_{k+1}|(T_{k+1}+s)^2\delta,\forall s>0,$$
which then complete our proof of the theorem.

By the definition of $E_{k}$, if no agents play any action in $E_{k}$, then  for any action $a\in E_{k+1}\setminus E_{k} $ of agent $i$ there must exist a mixed strategy $x$ of agent $i$ that dominates $a$. Assuming \model{} has run $T_k+T_0$ steps, we derive a sufficient condition for $T_0$ such that \model{} can $\varepsilon$-essentially eliminate actions in $E_{k+1}$ starting for any $t >T_k+T_0$. In particular, we show that for sufficiently large $T_0$ and any $s>0$, with probability at least $1-2|E_{k+1}|(T_k+T_0+s)^2\delta$, actions in $E_{k+1}$ are $\varepsilon'$-essentially eliminated from steps $T_k$ to $T_k+T_0+s$.

First of all, for any $a \in E_{k+1} \setminus E_k $ that is iteratively dominated by some $x$ and any $t \in (T_k,  T_k+T_0+s]$, from Lemma \ref{coro:gapy} we conclude with probability $1-(|E_{k+1}|-|E_{k}|)(T_0+s)\delta$, it holds that
\begin{align}\label{eq:gap_estm1} \notag 
& y_{x}(t+1)-y_{a}(t+1) \\ \notag 
 \geq &   \sum_{s=1}^t \gamma_s^{(t)}[u_{x}(s)-u_{a}(s)] -   4\left(\sqrt{\log\frac{2K}{\delta}}\cdot\sqrt{K(1+\sigma^2)\sum_{s=1}^t \frac{(\gamma_s^{(t)})^2}{\epsilon_s}}\right)  \\
 \geq &   \sum_{s=T_k+1}^t \gamma_s^{(t)}[u_{x}(s)-u_{a}(s)] - 2\sum_{s=1}^{T_k} \gamma_s^{(t)}  - 4\left(\sqrt{\log\frac{2K}{\delta}}\cdot\sqrt{K(1+\sigma^2)\sum_{s=1}^t \frac{(\gamma_s^{(t)})^2}{\epsilon_s}}\right),
 \end{align}
 where the last inequality used the bound $u_{x}(s)-u_{a}(s) \geq -2$ for any $s$ due to bounded utilities.

In order to further lower bound the RHS of Eq \eqref{eq:gap_estm1}, we need the following lemma:

\begin{lemma}\label{lm:learning_rate_property}
    If $T_0$ satisfies that
        \begin{equation}\label{eq:T0condition2}
        T_0\geq\max\Big\{\frac{(1+\beta)^2(4+\Delta)^2(8+\Delta)^2 }{4(1+2\beta)\Delta^2}\log \frac{1}{\delta}, \Big[(1+\frac{16}{\Delta})^{\frac{1}{1+\beta}}-1\Big]\cdot T_k \Big\},
    \end{equation} 
    then for any $s\geq 0, T'= T_0+s$ and $T=T'+T_k$, we have for any $a \in E_{k+1} \setminus E_k $, with probability at least $1-(|E_{k+1}| - |E_k|)T'\delta$,
        \begin{equation}\label{eq:lb_sumgammat} 
            \sum_{t=T_k+1}^{T_k+T'} \gamma_t^{(T)} [u_{x}(t)-u_{a}(t)] > \frac{\Delta}{4} \sum_{t=T_k+1}^{T_k+T'} \gamma_t^{(T)}.
        \end{equation}

\end{lemma}

The intuition behind Eq \eqref{eq:lb_sumgammat} is, since actions in $E_k$ are rarely played during $t\in(T_k, T_k+T_0+s]$, $u_x(t)-u_{a}(t)\geq \Delta$ must hold with high probability in this entire period. A standard technique to prove this is to use a union bound. However, this idea does not work here as it requires the undesirable event (i.e., $x_{-i}$ still contains certain essentially eliminated actions at a specific round) to happen with an extremely small probability $O(\frac{1}{T_k + T_0})$, which we cannot afford unless with an exponentially large $T_k$. To overcome this challenge, we take a different route and construct a sub-martingale regarding the utilities. The detailed proof of Lemma \ref{lm:learning_rate_property} is technical and we defer it to Appendix \ref{append:converge}. 

Now let $T_{k+1}=T_k+T_0$. The probability of Eq \eqref{eq:lb_sumgammat} to hold for any $T'\in[T_0, T_0+s]$ is thus at least $(1-(|E_{k+1}| - |E_k|)\delta\sum_{t=T_0}^{T_0+s}t )$. Substitute Eq \eqref{eq:lb_sumgammat} into Eq \eqref{eq:gap_estm1} and apply a union bound, we conclude that with a probability at least 
\begin{align*}
&1-2|E_k|(T_k+T_0+s)^2\delta-(|E_{k+1}|-|E_{k}|)(T_0+s)\delta-(|E_{k+1}| - |E_k|)\delta\sum_{t=T_0}^{T_0+s}t \\\geq & 1-2|E_{k+1}|(T_k+T_0+s)^2\delta=1-2|E_{k+1}|(T_{k+1}+s)^2\delta,
\end{align*}
the following inequality holds for any $T\in[T_{k+1}, T_{k+1}+s]$ (we omit the superscript $(T)$ on $\gamma_t$ in the following derivations):

\begin{align}
\notag
&y_{x}(T+1)-y_{a}(T+1) \\
&\geq \sum_{t=T_k+1}^{T} \gamma_t [u_{x}(t)-u_{a}(t)]- 2\sum_{t=1}^{T_k} \gamma_t - 4\left(\sqrt{\log\frac{2K}{\delta}}\cdot\sqrt{K(1+\sigma^2)\sum_{t=1}^T \frac{\gamma_t^2}{\epsilon_t}}\right) \\  \notag
&\geq  \left(\frac{\Delta}{8} \sum_{t=T_k+1}^{T} \gamma_t - 2\sum_{t=1}^{T_k} \gamma_t\right) + \left(\frac{\Delta}{8} \sum_{t=T_k+1}^{T}\gamma_t  - 4\left(\sqrt{\log\frac{2K}{\delta}}\cdot\sqrt{K(1+\sigma^2)\sum_{t=1}^T \frac{\gamma_t^2}{\epsilon_t}}\right)\right) \\ \label{eq:boundT1}
& \geq 0 + \left(\frac{\Delta}{16} \sum_{t=T_k+1}^{T}\gamma_t + \frac{1}{2} \sum_{t=1}^{T_k} \gamma_t  - 4\left(\sqrt{\log\frac{2K}{\delta}}\cdot\sqrt{K(1+\sigma^2)\sum_{t=1}^T \frac{\gamma_t^2}{\epsilon_t}}\right)\right) \\  \label{eq:boundT12}
& > \frac{\Delta}{16} \sum_{t=1}^{T}\gamma_t - 4\left(\sqrt{\log\frac{2K}{\delta}}\cdot\sqrt{K(1+\sigma^2)\sum_{t=1}^T \frac{\gamma_t^2}{\epsilon_t}}\right) \\ \label{eq:boundT13}
& > \frac{\Delta t}{16(1+\beta)}-4\sqrt{\left(\frac{eK(1+\sigma^2)}{1+2\beta+b}\right)\log\frac{2K}{\delta}}\cdot t^{\frac{1+b}{2}} ,
\end{align}
where Eq \eqref{eq:boundT1} holds because it is straightforward to verify $\frac{\Delta}{16} \sum_{t=T_k+1}^{T} \gamma_t^{(T)} - \sum_{t=1}^{T_k} \gamma_t^{(T)} \geq 0$  when $T_0\geq \left[(1+\frac{16}{\Delta})^{\frac{1}{1+\beta}}-1\right]\cdot T_k$, Eq \eqref{eq:boundT12} holds because $1 \geq \Delta > \frac{\Delta}{8}$, and Eq \eqref{eq:boundT13} holds because of Lemma \ref{coro:gapy}. Therefore, from Eq \eqref{eq:302} and Eq \eqref{eq:133} and the fact $T_k>T_1$ we conclude that all the actions in $E_{k+1}$ will be $\varepsilon$-essentially eliminated during round $t\in(T_{k+1}, T_{k+1}+s]$ with probability $1-2|E_{k+1}|(T_{k+1}+s)^2\delta$ as long as $T_0$ satisfies Eq \eqref{eq:T0condition2}. By induction, we complete the proof.

\section{Empirical Evaluations}\label{append:exp}


\paragraph{Baselines} We compare \model{} with a rich collection of basic or state-of-the-art online learning algorithms listed below: 
\begin{enumerate}[label=(\alph*)]
    \item   The classical Exp3 algorithm; 
    \item  Exp3.P which has no regret with high probability \cite{auer2002nonstochastic}; 
\item OMWU is the optimistic variant of MWU to obtain faster convergence to coarse correlated equilibria in games~\cite{daskalakis2021near};
\item  BM-OMWU applies the black-box reduction technique of Blum and Mansour (BM)~\cite{blum2005external} onto OMWU, which guarantees faster convergence to $\epsilon-$CE~\cite{anagnostides2021near};\footnote{The original algorithms in the paper are designed for full information feedback setting. To ensure fair comparison, we modify the OMWU and BM-OMWU to use the standard bandit estimator in EXP3. }
\item  The online mirror descent algorithm with log barrier regularizer, OMD-LB~\cite{foster2016learning}, which also uses increasing learning rate schedule \cite{lee2020bias,agarwal2017corralling,bubeck2017kernel}.
\end{enumerate}  We let all learning agents follow the same type of learning algorithm, i.e., self-play, in games described below and compare their convergence trend. 

\paragraph{Metrics} To measure the progress of iterative dominance elimination, we define the notion of \emph{elimination distance}~(ED) for any action $i$, $\Lambda(i)$, as the number of elimination iterations needed before this action start to be eliminated. Formally, $\Lambda(i) \equiv \arg\max_{0 \leq l \leq L_0}\{ i \not \in E_l \}$, where $E_0 = \varnothing$. 
The elimination distance of any undominated action is $L_0$, the elimination length (see Definition~\ref{def:length}).   
We let $\sum_{i\in \cA_n}x_n(i)\frac{\Lambda(i)}{L_0}$ be the normalized ED of mixed strategy $x_n$ of any agent~$n$.
We then introduce the \emph{Progress of Elimination} (PoE) metric, as the normalized elimination distance aggregated over all agents in the game  at  round $t$, 
$$
\textup{PoE}(t) = \frac{1}{N}\sum_{n\in \cN} \sum_{i\in \cA_n} p_{n,i}(t)\cdot \frac{\Lambda(i) }{ L_0 } \in [0, 1].
$$
Therefore, the larger PoE is, the more actions the learning agents have eliminated. When PoE reaches 1, the learning agents have removed all dominated actions and converged to the desirable set of rationalizable actions.

\paragraph{\gameshort{} game} To illustrate our theoretical results, we conduct a set of experiments in the \gameshort{} game, where all agents following the same type of learning algorithm. For DIR game, we use $ b=0.2, \beta=2K \approx L_0$ as the parameter of \model{}.
In Figure \ref{fig:exp2}, we can clearly observe that \model{} exhibits a superior performance in both cases, and enjoys a greater advantage in a larger game instance, where the other baselines even struggle to eliminate the first few dominated actions. In addition, OMD-LB with increasing learning rate also displays relatively good performance especially in the larger and harder instance, compare to other learning algorithms with non-increasing learning rate. But in the harder instance of DIR game with just $20$ actions on the right, all the baseline algorithms can only do elimination for 5 or 6 iterations out of $2K - 2 = 38$ iterations.

\begin{figure}[t]
    \centering
    \includegraphics[width=0.48\textwidth]{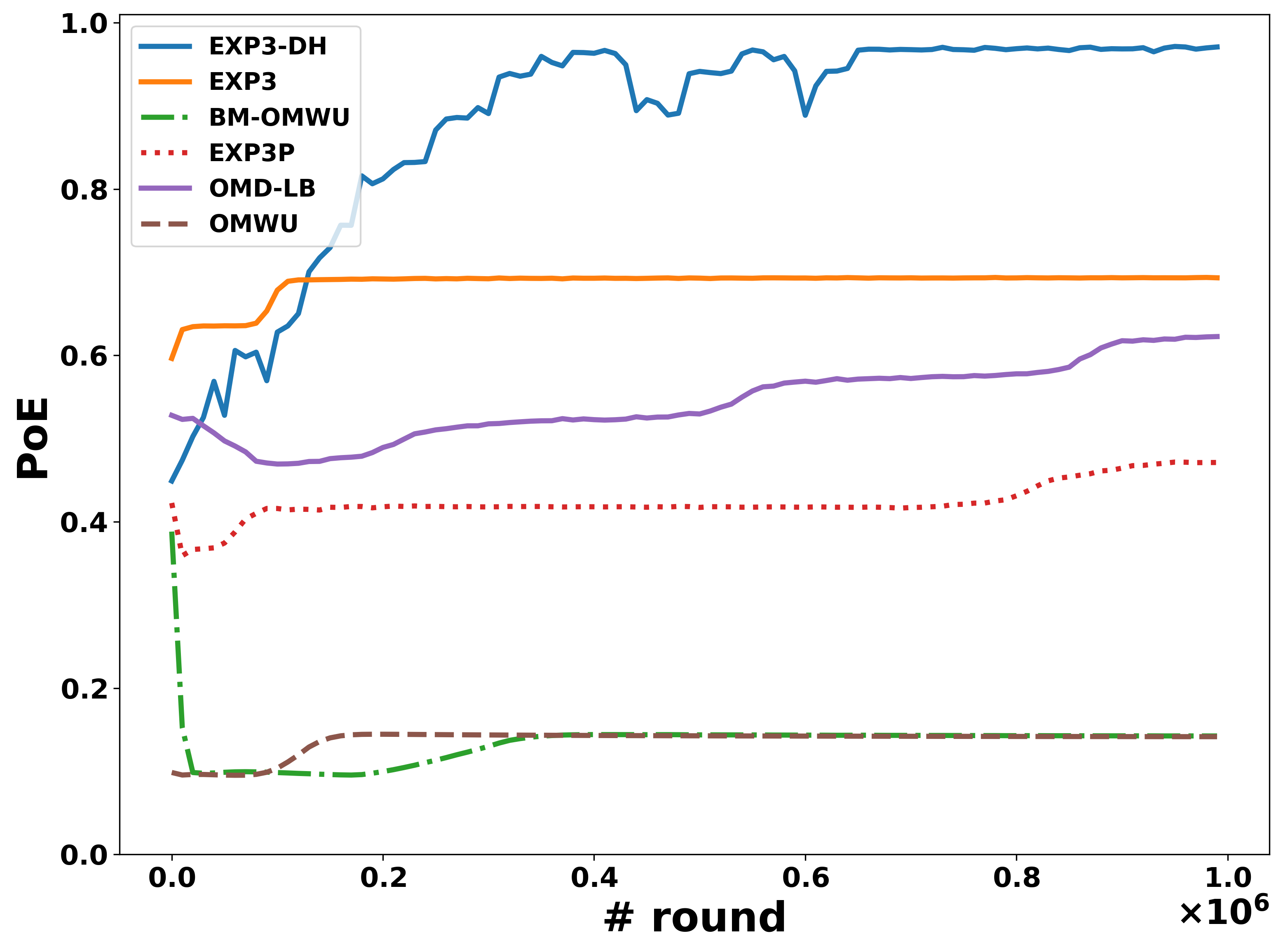}
    \includegraphics[width=0.48\textwidth]{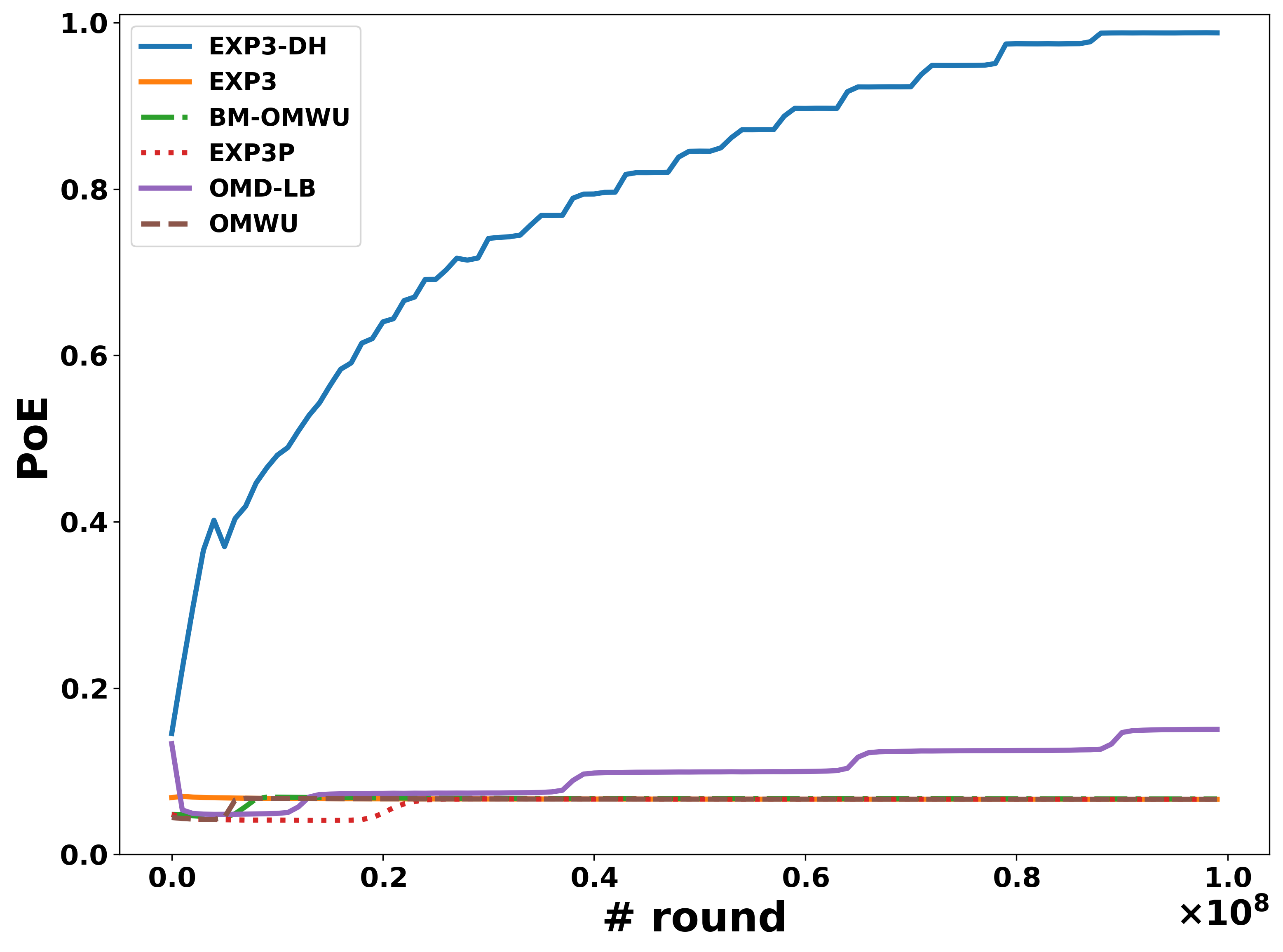}
    \caption{ Progress of Elimination (PoE) in a smaller  \textbf{ \gameshort$(10,20)$ game (left)} over $T=10^6$ rounds  and  a larger \textbf{\gameshort$(20,40)$ game (right)} over $T=10^8$ rounds.  
    In both games, i.i.d. Gaussian noise with std. $0.1$ is added onto agents' payoffs. 
   The performance of \model{} is represented by blue solid line while five baseline algorithms are represented by other notations shown in the legend. 
}
    \label{fig:exp2}
\end{figure}

\paragraph{The Market for ``Lemons'' (see Section \ref{sec:prelim:example})}  We also examine the algorithm performance on this famous example of the adverse selection problem by \citet{akerlof1978market}.  
Our numerical experiments aim at testing how fast the market will collapse, if every seller $i$ have noisy and bandit perception $\tilde{q}_i$ of their car's true quality.  In   our experiment instances, we set $c_1=3, c_2=1.5$, i.i.d. noise $\epsilon\sim \mathcal{N}(0, 5)$. Respectively, we choose $N=K-1=50$ or $200$ and let each $q_i = N/2+i$, $\mathcal{P} = \{N/2, \cdots, 3N/2\}$. We let the buyer and sellers apply no-regret learning algorithm to learn the equilibrium price and listing decisions from only the noisy bandit feedback in a repeated game. For \model{}, we set $ b=0.5, \beta \approx L_0$, since $N$ is of the same order with $K$ in this game. In Figure \ref{fig:exp}, as predicted by our theoretical results, with all learning agents running \model{} algorithm, the convergence can be polynomial, whereas the convergences from existing no-regret learning algorithms are comparatively slow. 

\begin{figure}[t]
    \centering
    \includegraphics[width=0.46\textwidth]{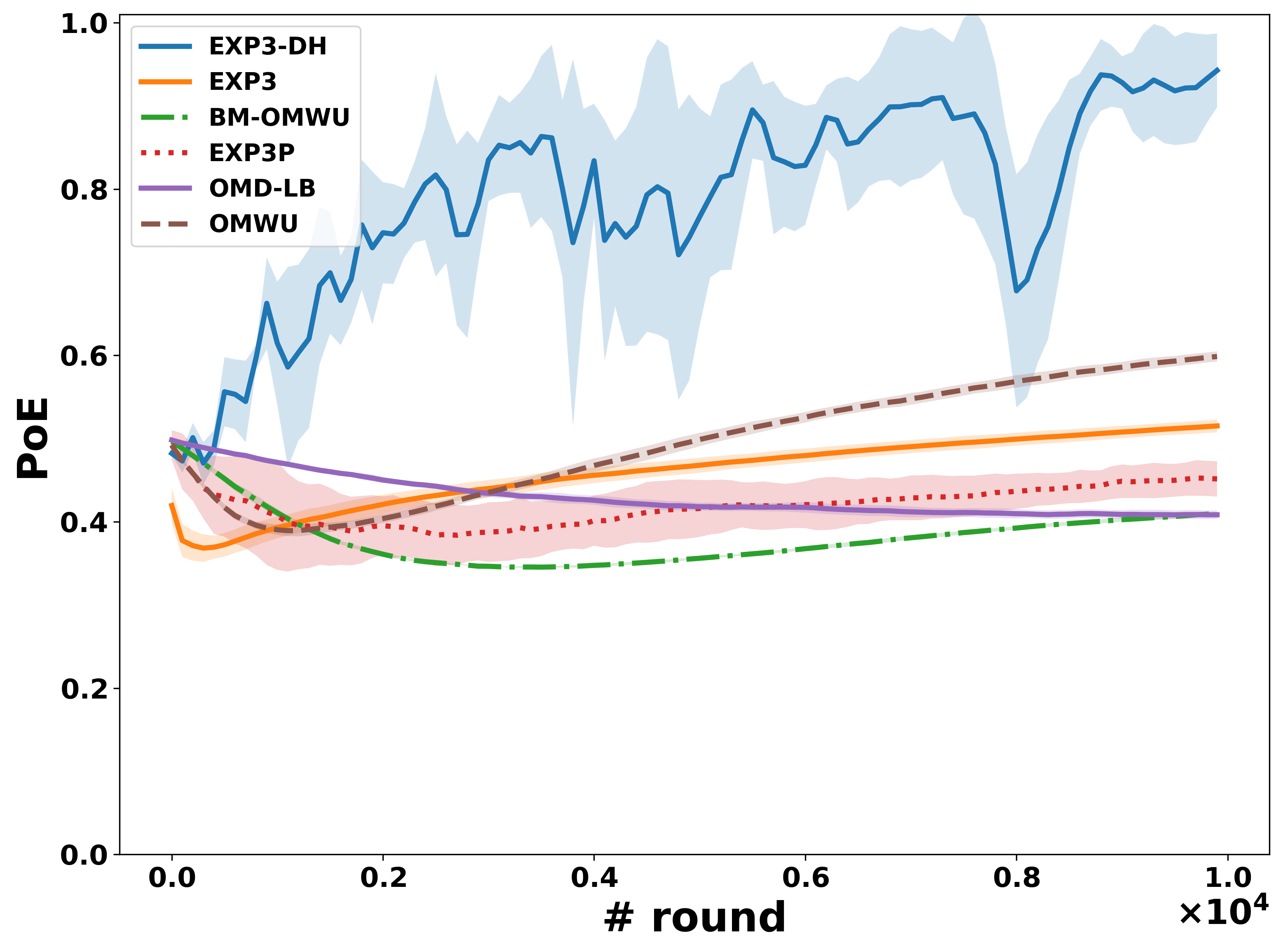}
    \qquad
    \includegraphics[width=0.46\textwidth]{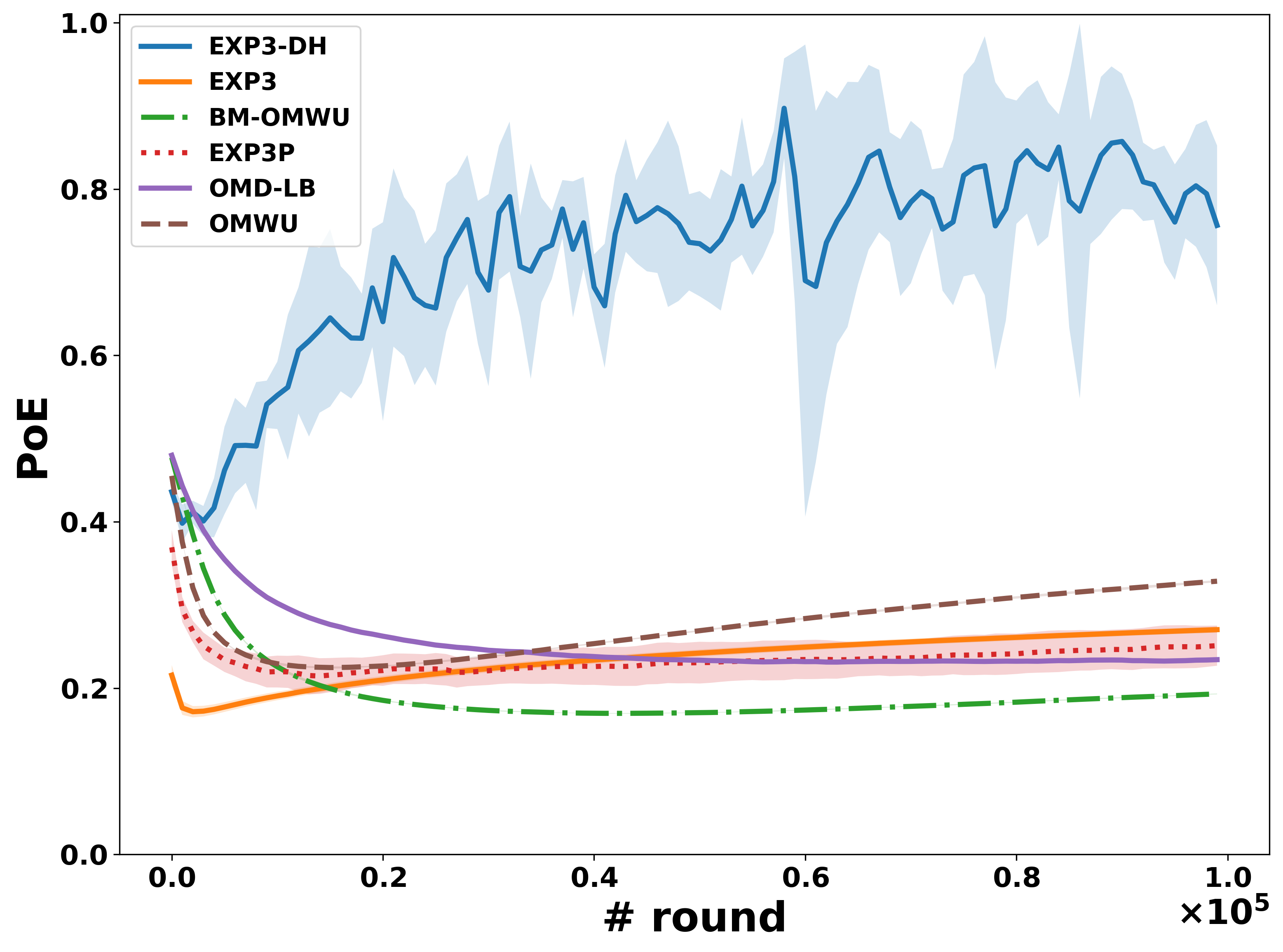}
    \caption{ Progress of Elimination (PoE) in \textbf{ The Market for ``Lemons'' } with 50 sellers (Left) or 200 sellers (Right).
    In this game, i.i.d. Gaussian noise with std. $0.1$ is added onto agents' payoffs. The lightly shaded region displays the error bar of each convergence trend (by one standard deviation over 5 runs).}
    \label{fig:exp}
\end{figure}

 
\paragraph{Performance in Presence of Adversarial Agents} \label{sec:regret}
We would also like to investigate the robustness of \model{} beyond the multi-agent learning setup. We conduct test the learning algorithms in two setups: one is to interface with the adversarial opponent in the \gameshort{} game that randomly plays a fixed action for every $1000$ rounds, the other is to run in the non-stationary environment where the arms' reward distributions change every $1000$ rounds. Such a switching frequency of $1000$ is particularly chosen for the adversarial purpose such that the learning algorithms are not best responding to a completely random reward distribution nor adapting to a periodic distribution change.

We plot the average regret incurred different algorithms in these two setups in Figure \ref{fig:exp-regret}.
To showcase the robustness of \model{} over the iterative best response approach by \cite{wang2022learning}, we also implement the $\epsilon$-greedy (EPSGreedy) algorithm and the explore-then-commit (ETC) algorithm that both periodically explores and commit to the best arm.\footnote{The exact algorithm in \cite{wang2022learning} is specified for a ``centralized'' learning setup (i.e., asking one agent to explore while the remaining agents follow the same action profile) and cannot be directly used for this experiment under the uncoupled learning setup. The $\epsilon$-greedy and ETC algorithms are arguably the closest variants of their iterative best response approach.} As demonstrated by the empirical results in Figure \ref{fig:exp-regret}, both ETC and EPSGreedy suffer in these adversarial environments, limiting their applications to well-specified multi-agent learning setups --- almost linear mean regret with large variance in both setups.  In contrast, the proposed \model{} shows strong performance in both settings, matching or even outperforming the Exp3 algorithm, which has the provable no-regret guarantee in adversarial settings. That said, we are still able to construct special instances where even \model{} suffers linear regret, and this points us towards an important open direction to design algorithms with better robustness guarantee in adversarial settings. 
  

\begin{figure}[t]
    \centering
    \includegraphics[width=0.46\textwidth]{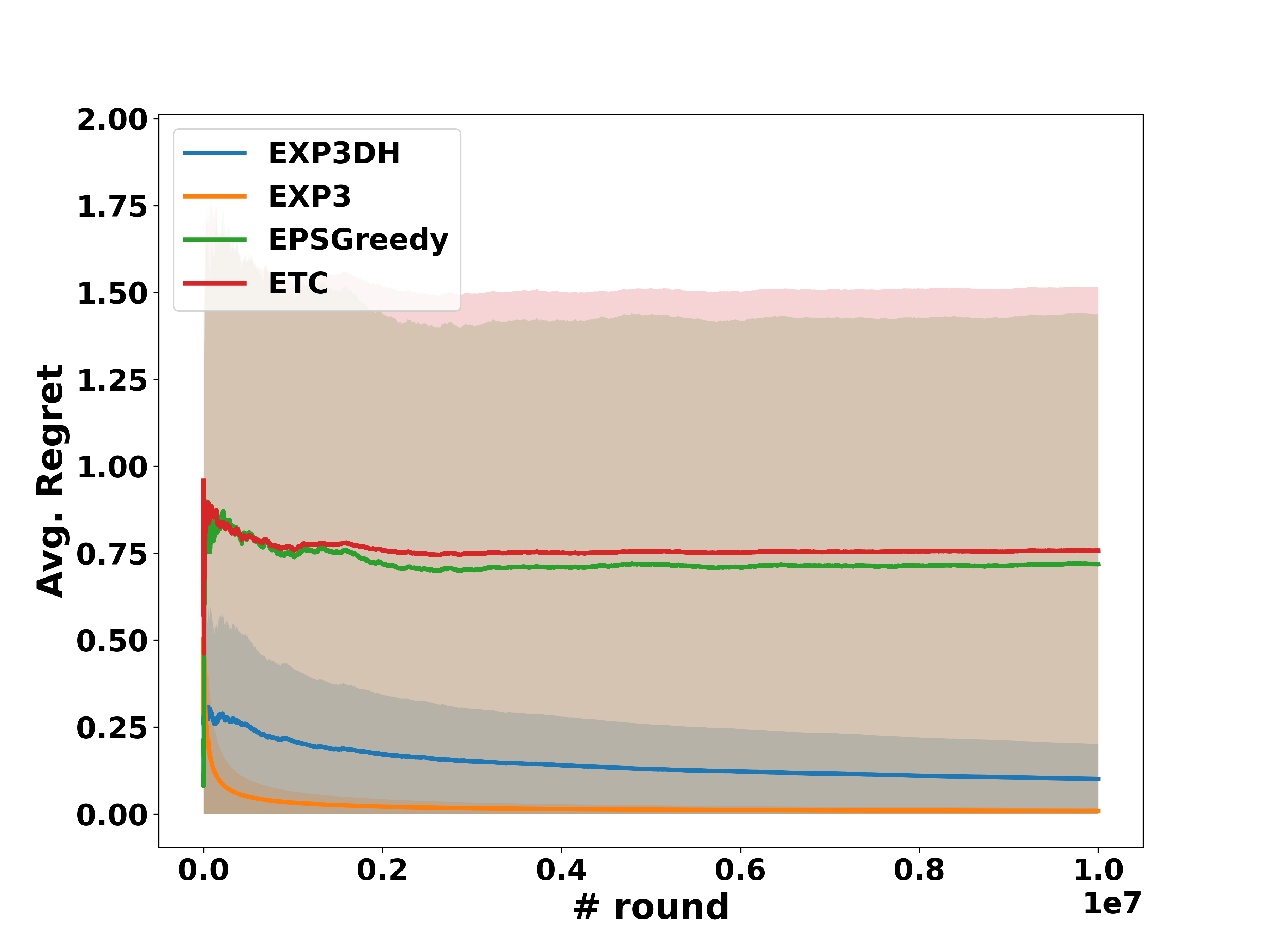}
    \qquad
    \includegraphics[width=0.46\textwidth]{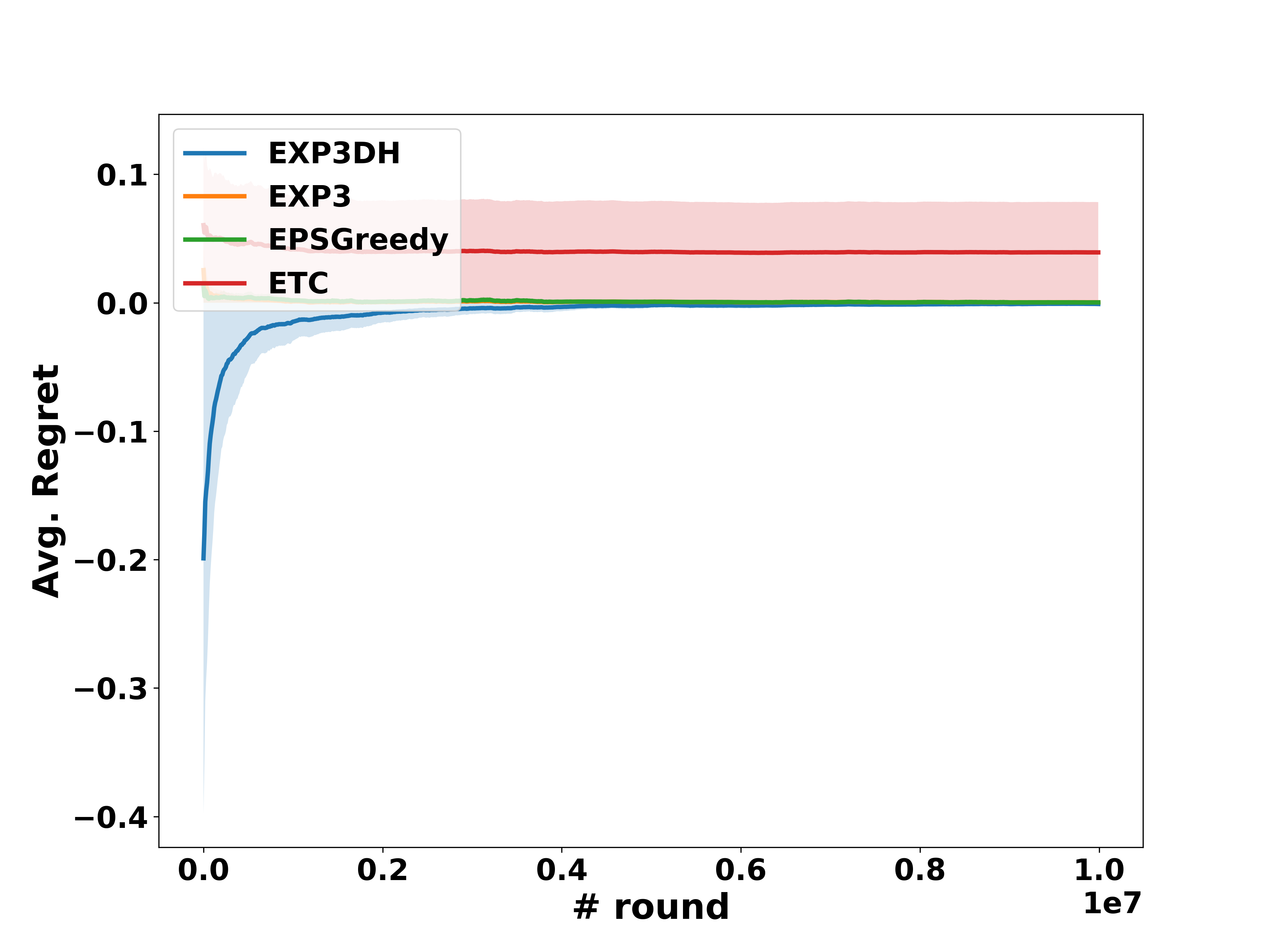}
    \caption{The history of average regret facing adversarial opponent in \textbf{\gameshort$(10,20)$} game (Left), in non-stationary environment (Right).
    In these games, i.i.d. Gaussian noise with std. $0.1$ is added onto agents' payoffs. The lightly shaded region displays the error bar of each convergence trend (by one standard deviation over 10 runs).}
    \label{fig:exp-regret}
\end{figure}

\section{Conclusion}
Our study formalizes the price of ``over-hedging'' of standard no-regret learning algorithms in the process of iterated dominance elimination. Such price is especially expensive in games whose equilibria are hidden in the ``rough''. In order to overcome this pitfall, we design a diminishing-history mechanism that deliberately balance the exploitation of the existing knowledge and indifference to history. However, there are still several open questions that remain for future works. 

 One direction is to understand the lower bound of convergence rate and whether there exists uncoupled learning algorithms with provably faster convergence guarantee.  In addition, \citet{wang2022learning} propose a learning algorithm with better convergence rate yet relying on agent's coordination, it is also interesting to understand the trade-off between the communication bandwidth and convergence rate. 
Moreover, despite we have proved the learning barrier of the \algoclass{} algorithm, it is unclear whether online mirror descent (OMD), another family of no regret learning algorithm, suffers the similar learning barrier.\footnote{We conjecture the answer is negative, at least for a subset of regularizers within OMD that we have already excluded: note that the DA class and OMD class overlap as EW can be interpreted as an OMD method with an entropic regularizer. Additionally, in Appendix~\ref{append:exp}, we empirically show that OMD with log-barrier regularizer \cite{foster2016learning} also suffers from exponentially slow convergence. } 
Finally, another important direction is to design the ``best of both worlds'' algorithm such that the rate is optimal both in regret and in last-iterate convergence to equilibrium. This raises a fundamental question on the existence of an intrinsic gap between the single-agent no-regret learning objective and the  multi-agent equilibrium learning.


\setcitestyle{numbers} 
\bibliography{ref}

\appendix

\newpage
 
\section{Supermodular Games}\label{append:supermodular}
The formal definition of supermodular games is based on several basic concepts from the lattice theory: For a partial order set $\L$ with operator $\geq$, we define the ``join'' operator $x\vee y$  for every $x,y\in\L$, as the least upper bound of $x, y$, i.e., $(x\vee y) \geq x$, $(x\vee y) \geq y$ and $z \geq (x\vee y)$ for all $z\in \L, z \geq x$ and $z \geq y$. Similarly, the ``meet'' operator $x\wedge y$ is defined for every $x,y\in\L$, as the greatest lower bound of $x, y$, i.e., $x \geq (x\wedge y)$, $y \geq (x\wedge y) $ and $(x\vee y) \geq z $ for all $z\in \L, x \geq z$ and $y \geq z$. 
$\L$ is a lattice if for every $x,y\in \L$, $(x\wedge y), (x\vee y)\in \L$.
We say $\L$ forms a complete lattice if every subset $\L' \subseteq \L$ forms a lattice.
Therefore, any compact subset of $\RR$ (with the usual order) is a complete lattice, as is any subset in $\RR^d$ formed by the product of $d$ compact sets in $\RR$ (with the product order). 

In addition, we say a function $f: \L \to \RR$ is supermodular if for all $x,y \in \L$, it holds   $
    f(x \wedge y) + f(x \vee y) \geq f(x) + f(y).
    $
A function $f: \L_1 \times \L_2 \to \RR$ has increasing difference if for all $x' \geq x \in \L_1, y' \geq y \in \L_2$, it holds  $
    f(x', y') - f(x, y') \geq f(x', y) - f(x, y). 
    $
Now we can formally define supermodular games as follows.

\begin{definition}[Supermodular Games]\label{def:supermodular}
A strategic game $\cG(\cN, \{\A_n\}_{n=1}^N, \{u_n\}_{n=1}^N)$ is supermodular if for each $n\in \cN$,
\begin{enumerate}
    \item Each $\A_n$ is a complete lattice.
    \item $u_n(a_n, a_{-n})$ is upper semi-continuous in $a_n$ for each fixed $a_{-n}$, and it is continuous in $a_{-n}$ for each fixed $a_{n}$, and has a finite upper bound.
    \item $u_n(a_n, a_{-n})$ is supermodular in $a_n$ for each fixed $a_{-n}$.
    \item $u_n(a_n, a_{-n})$ has \emph{increasing difference} in $a_n$ and $a_{-n}$.
\end{enumerate}
\end{definition}
The first condition is trivial when the $\cA_n$ is a totally-ordered set (e.g., a compact set of real numbers). The second condition is trivial when the $\cA$ is discrete. The third condition is trivial when the $\cA_n$ is single-dimensional. Hence, for the finite normal-form game of interest in this paper, we only need to focus on the condition of increasing difference. 

Moreover, as pointed out by \cite{Milgrom1990rationalizability, etessami2019tarski}, the structural and algorithmic properties of supermodular game is preserved in a broader class of \emph{games with strategic complementarities}, which relaxes the third and fourth condition to depend only on ordinal information on the utility functions, i.e. how the utilities compare to each other rather than their precise numerical values. In particular, the third condition can be relaxed to quasi-supermodularity, where a function $f: \L \to \RR$ is quasi-supermodular if for all $x,y \in \L$, it holds  $f(x) \geq f(x \wedge y)  \implies f(x \vee y) \geq f(y).$ The fourth condition can be relaxed to the single-crossing condition, where a function $f: \L_1 \times \L_2 \to \RR$ is single-crossing if for all $x' \geq x \in \L_1, y' \geq y \in \L_2$, it holds  $
    f(x', y) \geq f(x, y) \implies f(x', y')  \geq f(x', y).$

\section{Omitted Proofs in Section \ref{sec:barr} }\label{append:barr}
\subsection{Proof of Proposition \ref{prop:diamprop}}

 
\paragraph{Proof of the first claim. } This follows an induction argument. First, it is easy to see that A's action 1 is dominated by her action 2 since $u_A(1,j) = 1/\rho < 2/\rho = u_A(2,j)$ for all B's action $j\in[K]$. Second, we claim no other pure action of A or B can be eliminated by any pure or mixed strategy. This is due to two reasons: 1. for any $i>1$, there exists $j=i-1$ such that $ u_1(i,j) = \frac{i}{\rho} \geq   \max_{i'} u_1(i',j) = \frac{i}{\rho} \geq u_1(x,j)$ for any $x\in \cX_{A}$; 2. for any $j\geq 1$, there exists $i=j$ such that $ u_2(i,j) = \frac{j}{\rho} \geq   \max_{j'} u_2(i,j') = \frac{j}{\rho} \geq u_2(i,x)$ for any  $x\in \cX_{B}$.

Now suppose that for some $i\geq 1$, A has eliminated action $1,\cdots, i$ and B has eliminated action $1,\cdots, i-1$. We claim that B's action $i$ will be the only iteratively dominated action currently. Specifically, it is now iteratively dominated by her action $i+1$ because $u_B(j,i) = i/\rho < (i+1)/\rho = u_B(j,i+1)$ for all $j > i$ (note that A has eliminated any action $j\leq i$). Therefore, B will next eliminate action $i$. On the other hand, it is easy to verify that any action $i' \geq i+1$ of B is not dominated currently by any non-eliminated mixed or pure strategy since whenever A plays $i'$, B's action $i'$ yields the largest utility for her. Specifically, we have $u_B(i', i') = i'/\rho > u_B(i', j)$, which is $j/\rho$ when $j < i'$ and $-c/\rho$ when  $j > i'$. Therefore, any $i' \geq i+1$ cannot be a iteratively dominated action of B. Similarly, one can also verify that A's any action $i' \geq i+1$ currently is not a iteratively dominated strategy neither. So B's action $i$ is the only iteratively dominated action. 

Similar to the analysis above, one can show that after agent B's action $i$ is eliminated, agent A's action $i+1$ will become the only iteratively dominated action. By induction, we know the iterative elimination procedure will eliminate one action at each iteration until we reach the last action profile $(K,K)$. So the elimination length is exactly $2K-2$. 

 \paragraph{Proof of the second claim. } We show a stronger result: any finite dominance solvable game has a unique correlated equilibrium (CE). This   also implies the uniqueness of Nash equilibrium (NE) since any NE must be a CE as well. The above conclusion should be standard. However, since we are unable to find any existing proof for the reference, we here include a formal argument for the readers with little background in game theory for completeness of the paper. 
 
 Given any game $\cG$ with elimination length $L_0$ and elimination sets $E_1 \subset E_2 \cdots \subset E_{L_0}$, we claim that any CE must assign zero probability on actions in $E_{L_0}$. We complete the proof by contradiction. Suppose it is not the case; there must exist an action in $E_{L_0}$ that has positive probability in some CE. Let $l$ be the smallest index such that $E_l$ contains an action $a_n$ for agent $n$ in this CE. By our choice of $l$, all actions in $E_{l-1}$ have zero probability in the CE. This, however, implies that $a_n$ must be iteratively dominated by some other action $a'_n$ in this CE, which contradicts the definition of CE since whenever action $a_n$ is recommended to agent-$n$, he would strictly prefer to deviating to action $a'_n$. Therefore, any action in $E_{L_0}$ must have zero probability in any CE. This implies that there is only one unique CE in any finite dominance solvable game, which is also the unique NE. 
 

Finally, it is easy to see that $\frac{K}{\rho}$ is the largest possible utility in the payoff matrices. Therefore, at this equilibrium $(K,K)$, both agents achieves the maximum possible utility $\frac{K}{\rho}$, and the game obtains the maximum social welfare $\frac{2K}{\rho}$.

\subsection{Proof of Theorem \ref{thm:diamwelfare}}

\begin{proof}
We explicitly construct a mixed strategy for the two agents and prove that it constitutes an $\epsilon$-correlated-equilibrium and moreover satisfies the claimed properties in the stated theorem. Such construction is possible due to the special structure of the DIR games. Our construction is divided into two main steps: (1) constructing the support of the CE; (2) constructing the concrete probabilities. 

\noindent {\bf Support of the constructed CE. } Specifically, consider a distribution $\pi$ over the joint action space $[K] \times [K]$ with support in the following format  
$$\pi = \begin{bmatrix} 
\delta_1 & 0 & 0 & 0 & 0  \\
\delta_2 & \delta_3 & 0 & 0 & 0  \\
0 & \delta_4 & \delta_5 & 0 & 0  \\
0 & 0  & \cdots  &  \cdots &  0 \\
0 & 0 & 0 & \delta_{2K-2} & \delta_{2K-1}\\
\end{bmatrix}.$$
We claim that $\pi$ will be an $\epsilon$-CE if its  $\{\delta_i\}_{i=1}^{2K-1}$ satisfy the following inequality system:
\begin{equation}\label{eq:ce-constraints}
\begin{cases}
     \delta_1 \cdot 1 \leq \epsilon \cdot \rho \\ 
     \delta_2\cdot (-c-2)+\delta_3 \cdot 1 \leq \epsilon \cdot \rho \\  
     \delta_4\cdot (-c-3)+\delta_5 \cdot 1 \leq \epsilon \cdot \rho \\  
     \qquad \vdots \\
     \delta_{2K-4}\cdot (-c-(K-1))+\delta_{2K-3} \cdot 1 \leq \epsilon \cdot \rho 
\end{cases} 
\quad
\begin{cases}
     \delta_1\cdot (-c-1)+\delta_2 \cdot 1 \leq \epsilon \cdot \rho \\  
     \delta_3\cdot (-c-2)+\delta_4 \cdot 1 \leq \epsilon \cdot \rho \\  
     \qquad \vdots \\
     \delta_{2K-3}\cdot (-c-(K-1) )+\delta_{2K-2} \cdot 1 \leq \epsilon \cdot \rho
\end{cases}  
\end{equation}
and moreover $\sum_{i=1}^{2K-1} \delta_i = 1$. 
To see this,  by the structure of $\pi$, we know that whenever agent A is recommended action $i$, agent B will be recommended either action $i-1$ or action $i$. Since $u_A(i,i-1) = u_A(i,i) = i/\rho$,   agent A will get utility $i/\rho$ for sure when recommended action $i$. It is thus easy to see that agent A will not be willing to deviate to any action $i' < i$ since his utility can only be at most $i'/\rho$ for playing $i'$. Similarly, since $u_A(i', i) = u_A(i', i-1) = -c/\rho$ for any $i' > i+1$, so agent A will not deviate to any action $i' > i+1$ neither. In other words, the structure of $\pi$ already guarantees that whenever agent A's action $i$ is recommended, she would only be possibly having incentive to deviate action $i+1$, for which she gets utility $-c\delta_{2i-2}/\rho + (i+1)\delta_{2i-1}/\rho$. The $\epsilon$-CE condition thus requires the following  $$i(\delta_{2i-2} + \delta_{2i-1})/\rho \leq -c\delta_{2i-2}/\rho + (i+1)\delta_{2i-1}/\rho + \epsilon.$$
Simple algebraic calculation shows that the above is exactly the $i$th constraint in the left-hand-side of System \eqref{eq:ce-constraints}. Similarly, the $j$'th constraint in the right-hand-side of System \eqref{eq:ce-constraints} says whenever the agent B is recommended action $j$, she would prefer $j$ over $j+1$ for $j=1,\cdots, K-1$, which is the only constraint needed to guarantee $\epsilon$-CE. 


 \noindent {\bf Distribution of the constructed CE. } We now explicitly construct  $\{\delta_i\}_{i=1}^{2K-1}$ that satisfies linear system \eqref{eq:ce-constraints}. Suppose $c>1$.  Pick any $k \in \{2,\dots, 2K-1\}$, for any $\epsilon$ such that $1/\epsilon\in (\rho \sum_{i=1}^{k-1}c^{i-1
} , \rho \sum_{i=1}^{k}c^{i-1}  ]$, we can verify that 
\begin{equation} \label{eq:dir-eps-equilibrium}
    \delta_i= \begin{cases}
     \rho \cdot \epsilon c^{i-1}   & i < k \\
     1-\rho \cdot \epsilon \sum_{i=1}^{k-1}c^{i-1} & i = k \\
     0 & i > k\\
\end{cases} 
\end{equation} 
is a feasible solution to linear system \eqref{eq:ce-constraints}, and therefore forms an $\epsilon-$CE. With $k \leq 2K-2$, we obtained an $\epsilon-$CE with $\delta_{2K-1} = 0$ for any $\epsilon$ satisfying $1/\epsilon \leq c^{2K-2} \leq \rho  \sum_{i=1}^{2K-2}c^{i-1} $.  This concludes our proof for the first part of Theorem \ref{thm:diamwelfare}. 

\noindent {\bf Welfare property of the constructed CE.} Finally, we   derive the upper bound  of the welfare at the above $\epsilon-$CE. Observe that by construction of $\pi$, the welfare (i.e., sum of agents' utilities) at the action profile for $\delta_i$ is precisely $(i+1)/\rho$. Therefore, we can compute and bound the total welfare as follows:  
$$
    \sum_{i=1}^{2K-1} \delta_i \cdot \frac{i+1}{\rho} 
    = \sum_{i=1}^{k} \delta_i \cdot \frac{i+1}{\rho}
    \leq \sum_{i=1}^{k} \delta_i \cdot \frac{k+1}{\rho} 
    \leq \frac{k+1}{\rho}
$$
where the first equality is by Eq. \eqref{eq:dir-eps-equilibrium}, and  the first and second inequality is implied from the fact that $\frac{i+1}{\rho} \leq \frac{k+1}{\rho}, \forall i \leq k$ and $ \sum_{i=1}^{k} \delta_i = 1 $.


We know $\rho c^{k-2} \leq \rho \sum_{i=1}^{k-1}c^{i-1} <  1/\epsilon $. This implies $ k-2 <   \frac{\log(1/\epsilon)-\log(\rho)}{\log (c)}  $. 
Since $k\in \mathbb{N}$, we have $ k-1 \leq \lceil \frac{\log(1/\epsilon)-\log(\rho)}{\log (c)} \rceil $. Since $\rho \geq c$,  $\lceil \frac{\log(1/\epsilon)-\log(\rho)}{\log (c)} \rceil \leq \lceil \frac{\log(1/\epsilon)}{\log (c)} - 1 \rceil \leq \lceil \frac{\log(1/\epsilon)}{\log (c)} \rceil - 1 $. 
Therefore, the welfare is at most $  \frac{k+1}{\rho}$, which is at most $ \frac{1+ \lceil \log(1/\epsilon)/\log (c) \rceil}{2K} $  fraction of the equilibrium welfare $2K/\rho$. 

\end{proof}

\subsection{DIR Games Are not  Globally Variationally Stable}\label{append:variational}

According to the definition of variational inequality (Eq (6) in \citep{cohen2017hedging}), to see why the global variational inequality fails to hold, we only need to show that there exists $(x, y)\in \Delta_{[K]}\times \Delta_{[K]}$ such that $v(x,y)\cdot ((x, y)-(x^*,y^*))>-\frac{\Delta}{2}\|(x, y)-(x^*,y^*)\|$ for any $\Delta>0$, where $(x^*,y^*)=(0,\cdots,0,1,0,\cdots,0,1)$ is the unique NE of the game with payoff matrices defined in \eqref{eq:lowerboundinstance}, $\|\cdot\|$ is the $L_1$ norm, and $$v(x,y)=(v_{A,1}(y),v_{A,2}(y),\cdots,v_{A,K}(y),v_{B,1}(x),v_{B,2}(x),\cdots,v_{B,K}(x)),$$ where $v_{A,j}(y)$ is the payoff of when agent A plays pure strategy $j$ and agent B plays the mixed strategy $y$. 

Let $x=y=(0,\cdots,0,1,0,\cdots,0)$, i.e., the only non-zero term of $x$ and $j$ is the $i$-th element, where $1\leq i\leq K-1$. Then we have 
$$v(x,y)\cdot ((x, y)-(x^*,y^*)) = 2i > 0>-\frac{\Delta}{2}\|(x, y)-(x^*,y^*)\|,$$
meaning the global variational inequality is violated.

\section{Additional Discussion for \Algoclass{} Algorithms} \label{append:nonconverge}

In this section we demonstrate the broadness of \algoclass{} algorithms by showing all Dual Averaging (DA) and Follow the Perturbed-Leader (FTPL) algorithms are \algoclass{}. To complete the definition of \algoclass{} algorithms, we formally introduce the concept of order-preserving functions:

\begin{definition}\label{def:order-pre}
    We call a function $F: \mathbb{R}^K \rightarrow \Delta_K$ order-preserving if for any $\y=(y_1, \cdots, y_k)$ and $i,j$ such that $y_i<y_j$, we must have $F(\y)_i\leq F(\y)_j$. 
\end{definition}

We call an online learning algorithm \algoclass{} if it utilizes an order-preserving function to determine the action distribution from accumulated rewards. The formal definition is shown in Algorithm \ref{al:fair}. 

By definition \ref{def:fair_algo}, Any algorithm \ref{al:fair} equipped with an order-preserving function $F$ is \algoclass{}. During each round, a \algoclass{} algorithm first maps the accumulated score vector $\y_t$ to a distribution $\p_t=(p_1(t),\cdots,p_K(t))\in\Delta_K$ with a pre-chosen order-preserving function $F$ and then samples the current strategy from $\p_t$. A \algoclass{} algorithm also needs to specify a learning rate sequence $\{\eta_t\}$ to accumulate the collected rewards $\tilde{\u}_t$ from each round. We note that in general, the reward $\tilde{\u}_t$ can be any unbiased estimation of the true reward $\u_t$. However, since we focus on demonstrating a negative side of \algoclass{} algorithms in Theorem \ref{thm:nonconverge}, we consider the most informative type of feedback, i.e., the noiseless full-information setting. 

Next, we introduce the preliminaries of DA and FTPL algorithms, whose descriptions are shown in Algorithm \ref{al:da} and \ref{al:ftpl}.

\begin{algorithm}[h]
    \caption{The DA Algorithm Framework}
     \label{al:da}
     \begin{algorithmic}[1]
         \State \textbf{Input:} Mirror map $Q: \mathcal{Y}\rightarrow \Delta_K$, learning rate sequence $\{\eta_t>0\}$.

        \State $\y_1 \gets (0, \cdots, 0)$\;
        
        \For{$t = 1\dots T$}
            \State Compute $\p_t = Q(\y_t)$\;
            
            \State Draw an action $i_t$ from the distribution $\p_t$.\;
            
            \State Receive the expected payoff $\tilde{\u}_t=(\tilde{u}_1(t),\cdots,\tilde{u}_K(t)) $ for each action $i$ from the first-order oracle.\;
            
            \State Update $\y_{t+1} = \y_t + \eta_t \tilde{\u}_t$. \;
            
        \EndFor
    \end{algorithmic}
\end{algorithm}

In online learning literature, DA coincides with Follow-the-Regularized-Leader (FTRL) in the cases of linear losses, and is also known as the ``lazy'' version of online mirror descent (OMD)~\cite{shalev2011online}. Similar to the \algoclass{} defined in Algorithm \ref{al:fair}, the DA algorithm uses a ``mirror map'' $Q$ to derive the mixed strategy $\p_t$ from the accumulated reward $\y_t$. By convention, the mirror map $Q$ depends on a convex function (or regularizer) $h$ and is defined as
\begin{equation}\label{eq:mirror_map}
    Q(\y)=\arg\max_{x\in \Delta_K}\{\langle \y,\x \rangle-h(\x)\}, \y \in \mathcal{Y}.
\end{equation}

A regularizer is said to be \emph{symmetric} if for any $1\leq i< j\leq K$, 
 $ h(x_1, \dots, x_i, \dots, x_{j}, \dots, x_{K}) = h(x_1, \dots, x_j, \dots, x_i, \dots, x_{K})$. That is, the value of $h(x)$ will not change if we swap the values at any two coordinates of $x$. Most (if not all) known DA algorithms use symmetric and strictly convex regularizers. In this case, symmetry of $h$ implies a natural property of the algorithm --- i.e., if two actions have the same accumulated score, the algorithm should play either with equal probability. This exactly conforms to the definition of order-preserving function. We formalize this connection in Lemma \ref{lm:mono_Q}, which supports our main claim in Proposition \ref{prop:DAisfair}. 

\begin{proposition}\label{prop:DAisfair}
    Any DA algorithm equipped with a mirror map induced by a symmetric regularizer is \algoclass{}.
\end{proposition}

\begin{proof}
By the definition of \algoclass{} algorithms, we only need to show the following Lemma:
\begin{lemma}\label{lm:mono_Q}
Any mirror map $\p = Q(\y)$ induced by a symmetric regularizer $h$ is order-preserving.
\end{lemma}
\begin{proof}
We prove by contradiction. Suppose for some $\y$ and $i, j$ where $y_i > y_j$, while $\p = Q(\y)$ with $p_i < p_j$. By definition, $\p = Q(\y) =\arg\max_{x\in \Delta_K}\{\langle \y,\x \rangle-h(\x)\}$. However, consider the vector $\tilde{\p}= (p_1, \dots, p_j, \dots, p_i, \dots, p_{K})$. By construction, $\tilde{\p} \in \Delta_K$ and $h(\tilde{\p}) = h(\p) $. By rearrangement inequality, we have $\tilde{p}_i y_j + \tilde{p}_j y_j  = p_j y_i + p_i y_j > p_i y_j + p_j y_j $ such that $\langle \y,\tilde{\p} \rangle-h(\tilde{\p}) > \langle \y,\p \rangle-h(\p)$. This is a contradiction to $\p = \arg\max_{x\in \Delta_K}\{\langle \y,\x \rangle-h(\x)\}$. Therefore, it must be the case that whenever $y_i > y_j$, $p_i \geq p_j$. By definition \ref{def:order-pre}, Lemma \ref{lm:mono_Q} holds.
\end{proof}
\end{proof}

The DA family includes many celebrated algorithms such as Exponential Weight (EW), lazy gradient descent (LGD) and fictitious play. 
Below we list a few widely used algorithms from the DA family which, unsurprisingly, all use symmetric and strictly convex regularizers: 
\begin{enumerate}
    \item when $h(\x)=\frac{1}{2}\|\x\|^2$ is the quadratic function, the mirror map 
    $Q(\y)=\arg\min_{\x\in \Delta_K}\|\x-\y\|^2$ takes the form of Euclidean projection and we obtain the lazy gradient descent (LGD) algorithm. 
    \item when $h(\x)=\sum_{i\in[K]}x_i\log x_i$ is the entropic regularizer, the mirror map 
    $Q(\y)=\frac{(\exp(y_1),\cdots,\exp(y_K))}{\exp(y_1)+\cdots+\exp(y_1)}$ takes the form of logit choice map and we obtain entropic gradient descent algorithm, which is also known as the Hedge or Exponential Weight (EW).
    \item when $h(\x)=\frac{1}{p}\|\x\|^p$ is the normalized $L_{p}$ norm and $p\rightarrow \infty$, $Q(\y)$ always returns the pure best strategy to the opponent's average past mixed strategy, and the corresponding algorithm is known as fictitious play \cite{brown1951iterative, viossat2013no}.
\end{enumerate}

Next, we show that FTPL algorithms are also \algoclass{}. The outline of FTPL family is shown in Algorithm \ref{al:ftpl}.

\begin{algorithm}[h]
    \caption{The FTPL Algorithm Framework}
     \label{al:ftpl}
     \begin{algorithmic}[1]
         \State \textbf{Input:} A probability distribution $\D$, learning rate sequence $\{\eta_t>0\}$.

        \State $\y_1 \gets (0, \cdots, 0)$\;
        
        \For{$t = 1\dots T$}
            \State Sample $\{\epsilon_i\}_{i=1}^K$ independently from $\D$.\;
            
            \State Choose the action $i_t=\arg\max_{i\in[K]}\{y_i(t)+\epsilon_i$\}.\;
            
            \State Receive the expected payoff $\tilde{\u}_t=(\tilde{u}_1(t),\cdots,\tilde{u}_K(t)) $ for each action $i$ from the first-order oracle.\;
            
            \State Update $\y_{t+1} = \y_t + \eta_t \tilde{\u}_t$. \;
        \EndFor
    \end{algorithmic}
\end{algorithm}

The distribution $\D$ can be any single variable distribution, e.g., Gaussian, Gumbel, exponential, etc. When $\D$ takes the Gumbel distribution with zero mean, the FTPL algorithm coincides with multiplicative weights update (MWU). For FTPL algorithms, the link function $F$ that maps the accumulated rewards to action distribution is implicit. However, it is straightforward to see that $F$ is order-preserving and thus all FTPL algorithms are \algoclass{}. We formalize the claim in the following proposition \ref{prop:FTPLisfair}:

\begin{proposition}\label{prop:FTPLisfair}
    The class of FTPL algorithms are \algoclass{}.
\end{proposition}
\begin{proof}
Any FTPL algorithm can be equivalently represented in the form of Algorithm \ref{al:fair} with a link function $F$ defined as
$$F(\y)_i=\Prob{i=\arg\max_{k\in[K]}\{y_k+\epsilon_k\}}, \forall k \in [K].$$
Because random variables $\epsilon_i$ are i.i.d., for any $y_i>y_j$ it must hold that 
$$\Prob{i=\arg\max_{k\in[K]}\{y_k+\epsilon_k\}}\geq \Prob{j=\arg\max_{k\in[K]}\{y_k+\epsilon_k\}}.$$  Therefore, $F$ is order-preserving and the induced algorithm is \algoclass{}.
\end{proof}

\section{Omitted Proofs in Section \ref{sec:converge} }\label{append:converge}

\subsection{Proof of Corollary \ref{coro:maintheorem}}

From Eq \eqref{eq:302}, we know that a sufficient condition for $t$ to $\varepsilon$-\emph{essentially eliminate} a dominated action $a \in E_1$ is
\begin{equation}\label{eq:suff_ee2}
\frac{t^{-b}}{K} +  \exp\left(4\left(\sqrt{\frac{eK(1+\sigma^2)}{1+2\beta+b}}\right)\log^{\frac{1}{2}}\frac{2K}{\delta}\cdot t^{\frac{1+b}{2}} -\frac{\Delta t}{16(1+\beta)} \right)  < \frac{\min\{\varepsilon,\Delta/2\}}{4KN}.    
\end{equation}

Note that a sufficient condition to satisfy Eq \eqref{eq:suff_ee2} above is the following 
\begin{eqnarray*}
& &  4\left(\sqrt{\frac{eK(1+\sigma^2)}{1+2\beta+b}}\right)\log^{\frac{1}{2}}\frac{2K}{\delta}\cdot t^{\frac{1+b}{2}} -\frac{\Delta t}{32(1+\beta)} < 0 \\
& & \frac{t^{-b}}{K} < \frac{\min\{\varepsilon,\Delta/2\}}{8KN} \qquad \text{ and } \qquad -\frac{\Delta t}{32(1+\beta)}  < \log \frac{\min\{\varepsilon,\Delta/2\}}{8KN},
\end{eqnarray*} 
which can be simplified to 
\begin{small}
\begin{equation*} 
t > \max\{O(N^{\frac{1}{b}}\min\{\varepsilon,\Delta\}^{-\frac{1}{b}}), O(K^{\frac{1}{1-b}}(1+\sigma^2)^{\frac{1}{1-b}}\beta^{\frac{1}{1-b}}\Delta^{-\frac{2}{1-b}}\log^{\frac{1}{1-b}}\frac{1}{\delta}),O(\beta\Delta^{-1}\log\frac{KN}{\min\{\varepsilon,\Delta\}}) \}.    
\end{equation*}
\end{small}
To satisfy the above condition, we can simply take  $$t=O\left(\max\{N^{\frac{1}{b}},K^{\frac{1}{1-b}}\}\max\{\varepsilon^{-\frac{1}{b}},\Delta^{-\max\{\frac{1}{b},\frac{2}{1-b}\}}\}(1+\sigma^2)^{\frac{1}{1-b}}\beta^{\frac{1}{1-b}}\log^{\frac{1}{1-b}}\frac{1}{\delta}\right).$$ 
In particular, when $b=\frac{1}{3}$, we may choose $T_1=O\left(\max\{N^3,K^{1.5}\}\min\{\varepsilon,\Delta\}^{-3}(1+\sigma^2)^{1.5}\beta^{1.5}\log^{1.5}\frac{1}{\delta}\right)$. 
Therefore, with probability at least $1-|E_1|(T_1+s)>1-2|E_1|(T_1+s)^2$, actions in $E_1$ will be essentially eliminated at round $T_1+1, \cdots, T_1+s$.

From Eq \eqref{eq:eq4Tk}, we can thus upper bound $T_{L_0}$ by
\begin{align*}
T_{L_0}&<(1+\frac{8}{\Delta})^{\frac{L_0}{1+\beta}}\cdot \left(T_1+L_0\cdot\frac{(1+\beta)^2(4+\Delta)^2(8+\Delta)^2 }{4(1+2\beta)\Delta^2}\log \frac{1}{\delta}\right)\\& \sim O\left(\Delta^{-\frac{L_0}{1+\beta}}(T_1+L_0\beta\Delta^{-2}\log\frac{1}{\delta})\right) \\ 
&= O\left(\Delta^{-\frac{L_0}{1+\beta}}(\max\{N^3,K^{1.5}\}\min\{\varepsilon,\Delta\}^{-3}(1+\sigma^2)^{1.5}\beta^{1.5}\log^{1.5}\frac{1}{\delta}+L_0\beta\Delta^{-2}\log\frac{1}{\delta})\right) \\ 
&=O\left(\max\{N^3,K^{1.5}\}\min\{\varepsilon,\Delta\}^{-3}(1+\sigma^2)^{1.5}\beta^{1.5}\log^{1.5}\frac{1}{\delta}\right),
\end{align*}
which implies that $E_{L_0}$ will be $\varepsilon$-essentially eliminated in $O\left(\max\{N^3,K^{1.5}\}\min\{\varepsilon,\Delta\}^{-3}(1+\sigma^2)^{1.5}\beta^{1.5}\log^{1.5}\frac{1}{\delta}\right)$ iterations with probability at least $1-2KNT_{L_0}^2\delta$. Hence, we claim that $E_{L_0}$ will be $\varepsilon$-essentially eliminated in $\tilde{O}\left(\max\{N^3,K^{1.5}\}\min\{\varepsilon,\Delta\}^{-3}(1+\sigma^2)^{1.5}\beta^{1.5}\log^{1.5}\frac{1}{\delta}\right)$ iterations with probability at least $1-\delta$ for any $\delta$.

\subsection{Proofs of Technical Lemmas}

\begin{proof}\textbf{of Lemma \ref{le:decompE}}

Fix $T>0$ and any action $a$. Let $\Xi_t=\gamma_t^{(T)} (\tilde{u}_a(t)-u_a(t))$ and $ S_t=\sum_{s=1}^t \Xi_s$. Since $T$ is fixed, $\mathbb{E}[S_t]$ is bounded. Moreover, we have
\begin{align*}
    \mathbb{E}[S_t|S_{t-1},\cdots,S_1] &= S_{t-1} + \gamma_t^{(T)} \cdot \mathbb{E}[\tilde{u}_{a}(t)-u_{a}(t)|\cF_{t-1}]=S_{t-1}.
\end{align*}
By definition, $\{S_t\}_{t=1}^T$ is a martingale.


Apply Azuma’s inequality to $\{S_t\}$, for any $x>0, 1\leq t \leq T$, we have  
\begin{equation}\label{eq:Benn}
    \mathbb{P}[|S_t| \geq x]\leq 2\exp{\left(-\frac{x^2}{2W}\right)}.
\end{equation}
$W$ is a upper bound of $\sum_{i=1}^t \mathbb{E}[\Xi_i^2|\cF_{i-1}]$. Note that
\begin{align}
\notag
\mathbb{E}[\Xi_t^2|\mathcal{F}_{t-1}]&= (\gamma_t^{(T)})^2\mathbb{E}\left[\left(0-u_{a}(t)\right)^2 \cdot (1-p_{a}(t)) + \left(\frac{u_{a}(t)+\xi_{t}}{p_{a}(t)}-u_{a}(t)\right)^2 \cdot p_{a}(t)\Big|\mathcal{F}_{t-1}\right] \\  \notag
&=(\gamma_t^{(T)})^2\mathbb{E}\left[(u^2_{a}(t)+\xi_t^2) \cdot \frac{1}{p_{a}(t)}-u^2_{a}(t) \Big|\mathcal{F}_{t-1}\right] \\ \label{eq:xisquareupper}
& \leq (\gamma_t^{(T)})^2\left(\frac{K(1+\sigma^2)}{\epsilon_t} \right),
\end{align}
we can take $n=T$ and $W=\sum_{i=1}^T \gamma_i^2(\frac{K(1+\sigma^2)}{\epsilon_i})$ in Eq \eqref{eq:Benn} and obtain
\begin{equation}\label{eq:solvex}
    x\geq \sqrt{2W\log \frac{2}{\delta}},
\end{equation}
which implies with probability at least $1-\delta$, 
$$|\sum_{t=1}^T \gamma_t^{(T)}(\tilde{u}_{a}(t)-u_{a}(t))| < 2\left(\sqrt{\log\frac{2}{\delta}}\cdot\sqrt{K(1+\sigma^2)\sum_{t=1}^T \frac{(\gamma_t^{(T)})^2}{\epsilon_t}}\right).$$

\end{proof}

\begin{proof}\textbf{of Lemma \ref{coro:gapy}}

By the definition of strict dominance, there exists $\Delta>0$ such that the mixed strategy $x=(x_1, \cdots, x_K)$ satisfies $u_{x}(t)-u_{a}(t) > \Delta$ for all $t>0$. From Lemma \ref{le:decompE}, we can take a union bound over all the actions $a\in [K]$ and obtain with probability $1-K\delta'$,
$$|\sum_{t=1}^T \gamma_t^{(T)}(\tilde{u}_{a}(t)-u_{a}(t))| < 2\left(\sqrt{\log\frac{2}{\delta'}}\cdot\sqrt{K(1+\sigma^2)\sum_{t=1}^T \frac{(\gamma_t^{(T)})^2}{\epsilon_t}}\right), \forall a\in [K].$$

As a result, 
\begin{align*}
& y_{x}(T+1)-y_{a}(T+1) \\ =&  \sum_{t=1}^T \gamma_t [\tilde{u}_{x}(t)-\tilde{u}_{a}(t)] \\ \notag
=& \sum_{t=1}^T \gamma_t [u_{x}(t)-u_{a}(t)]+ \sum_{t=1}^T \gamma_t [-u_{x}(t)+\tilde{u}_{x}(t)]+\sum_{t=1}^T \gamma_t [-\tilde{u}_{a}(t)+u_{a}(t)] \\ 
=& \sum_{t=1}^T \gamma_t [u_{x}(t)-u_{a}(t)]+ \sum_{i\in\cA }x_i\sum_{t=1}^T \gamma_t [-u_{i}(t)+\tilde{u}_{i}(t)]+\sum_{t=1}^T \gamma_t [-\tilde{u}_{a}(t)+u_{a}(t)] \\ 
\geq &  \sum_{t=1}^T \gamma_t [u_{x}(t)-u_{a}(t)]- 4 \left(\sqrt{\log\frac{2}{\delta'}}\cdot\sqrt{K(1+\sigma^2)\sum_{t=1}^T \frac{\gamma_t^2}{\epsilon_t}}\right).
\end{align*}
Letting $\delta'=\frac{\delta}{K}$ yields Eq \eqref{eq:deltay}. When $\gamma_t = (t/T)^{\beta}, \epsilon_t = t^{-b}$ and assume $T>\beta$, we have
\begin{equation}\label{eq:283}
    \sum_{t=1}^T \gamma_t [u_{x}(t)-u_{a}(t)]\geq\Delta \sum_{t=1}^T \gamma_t =\Delta \sum_{s=1}^T (\frac{s}{T})^{\beta} > \Delta T\cdot\int_{0}^{1} x^{\beta}dx = \frac{\Delta T}{1+\beta},
\end{equation}
\begin{align}
\notag
\sqrt{K(1+\sigma^2)\sum_{t=1}^T \frac{\gamma_t^2}{\epsilon_t}} & < \sqrt{K(1+\sigma^2)T^{1+b}\int_{\frac{1}{T}}^{1+\frac{1}{T}}x^{2\beta+b}dx}  \\ \label{eq:287}
& < \sqrt{\frac{eK(1+\sigma^2)}{1+2\beta+b}\cdot T^{1+b}},
\end{align}
where Eq \eqref{eq:287} holds because 
$$\int_{\frac{1}{T}}^{1+\frac{1}{T}}x^{\alpha}dx=\frac{1}{1+\alpha}\cdot\left[\left(1+\frac{1}{T}\right)^{1+\alpha}-\left(\frac{1}{T}\right)^{1+\alpha}\right] < \frac{1}{1+\alpha}\left(1+\frac{1}{T}\right)^{T}<\frac{e}{1+\alpha}.$$
Therefore, Eq \eqref{eq:deltay2} holds.
\end{proof}

\begin{proof}\textbf{of Lemma \ref{lm:learning_rate_property}}
Consider the sequence of events $$A_{t}=\{x_{-i} \text{~contains~essentially~eliminated~actions~at~round~}T_k+t\}, \quad  t=1,\cdots,T'.$$ 


Specifically, let $T=T_k+T'$ be fixed and define random variables $$Z_t=\sum_{s=T_k+1}^{T_k+t} \gamma_s^{(T)} [u_{x}(s)-u_{a}(s)-\frac{\Delta}{2}], \quad  1\leq t\leq T'.$$ We further let $Z_0=0$ and now show that $\{Z_t\}_{t=0}^{T'}$ is a sub-martingale. Since $E_k$ has been $\varepsilon$-essentially eliminated at any $t\geq T_k$ and since $\varepsilon<\frac{1}{2}$, we have
$$\mathbb{P}[A_t|A_{t-1},\cdots,A_1] < |E_k| \cdot \frac{\min\{\varepsilon,\Delta/2\}}{4KN} \leq KN\cdot\frac{\Delta}{8KN}=\frac{\Delta}{8}.$$ 

Note that when $A_t$ does not happen, no actions in $E_k$ will be played by agent $i$'s opponents, in which case we have  $u_{x}(t)-u_{a}(t)\geq \Delta$; On the other hand,  when $A_t$ happens we have $u_{x}(t)-u_{a}(t)\geq -2$ due to bounded utilities. Hence,
\begin{align*}
    & \mathbb{E}[Z_t|Z_{t-1},\cdots,Z_1] \\
    &= Z_{t-1} + \gamma_t^{(T)} \cdot \mathbb{E}[u_{x}(t)-u_{a}(t)]-\gamma_t^{(T)} \cdot\frac{\Delta}{2}\\
    &\geq Z_{t-1} + \gamma_t^{(T)} \cdot \left[\mathbb{P}[A_t|A_{t-1},\cdots,A_1] \cdot (-2)+(1-\mathbb{P}[A_t|A_{t-1},\cdots,A_1])\cdot \Delta\right]-\gamma_t^{(T)} \cdot \frac{\Delta}{2} \\ 
    & > Z_{t-1} + \gamma_t^{(T)} \cdot \left[\frac{\Delta}{8} \cdot (-2)+(1-\frac{\Delta}{8})\cdot \Delta\right] -\gamma_t^{(T)} \cdot \frac{\Delta}{2} \\  
    & > Z_{t-1} + \gamma_t^{(T)} \cdot \left[\frac{\Delta}{8} \cdot (-2)+(1-\frac{2}{8})\cdot \Delta -\frac{\Delta}{2}\right]=Z_{t-1}.
\end{align*}
Therefore, $\{Z_t\}_{t=1}^{T'}$ is a sub-martingale. Therefore, let $c_i=\gamma_{T_k+i}^{(T)}(2+\frac{\Delta}{2})$ denote a upper bound of $|Z_{i}-Z_{i-1}|$. By Azuma's inequality, we have 
\begin{equation}\label{eq:azuma}
    \mathbb{P}[Z_t \leq -\epsilon] \leq \exp \left(\frac{-\epsilon^2}{2\sum_{i=1}^t c_i^2}\right), \forall \epsilon >0.
\end{equation}
Let $$ t=T', \epsilon=\frac{\Delta}{4}\sum_{t=T_k+1}^{T_k+T'} \gamma_t^{(T)}, \text{~and~} \exp \left(\frac{-\epsilon^2}{2\sum_{i=1}^t c_i^2}\right) \leq \delta,$$ we obtain
\begin{equation}
\sum_{t=T_k+1}^{T_k+T'} \gamma_t [u_{a}(t)-u_{a'}(t)] > \frac{\Delta}{4} \sum_{t=T_k+1}^{T_k+T'} \gamma_t
\end{equation} 
with probability at least $1-(|E_{k+1}| - |E_k|)T'\delta$ for any $a \in E_{k+1} \setminus E_k $, as long as $T'$ satisfies 
\begin{equation}\label{eq:T0condition0}
\frac{\Delta^2}{16}\left(\sum_{t=T_k+1}^{T_k+T'}\gamma_t\right)^2 > 2 (2+\frac{\Delta}{2})^2\sum_{t=T_k+1}^{T_k+T'}\gamma_t^2\log \frac{1}{\delta}.
\end{equation}
We now derive a sufficient condition for Eq \eqref{eq:T0condition0} to hold, after substituting $\gamma_t^{(T)} = (t/T)^{\beta}$ into Eq \eqref{eq:T0condition0} and assuming $T\geq(1+\frac{8}{\Delta})^{\frac{1}{1+\beta}}T_k,T\geq T_k+1+2\beta$. We first lower bound the LHS of Eq \eqref{eq:T0condition0} via the following bound (recall $T=T_k+T'$):
\begin{align}\notag
\sum_{t=T_k+1}^{T_k+T'}\gamma_t^{(T)} &> T'\int_{\frac{T_k}{T}}^1x^{\beta}dx \\ \notag
&= \frac{T'}{1+\beta}\left[1-\left(\frac{T_k}{T}\right)^{1+\beta}\right] \\ \label{eq:434}
& \geq \frac{8T'}{(8+\Delta)(1+\beta )},
\end{align}
where Eq \eqref{eq:434} holds because $T\geq(1+\frac{8}{\Delta})^{\frac{1}{1+\beta}}T_k$. We then upper bound the RHS of Eq \eqref{eq:T0condition0} via the following bound: 
\begin{align}\notag
\sum_{t=T_k+1}^{T_k+T'}(\gamma_t^{(T)})^2 & < - \left(\frac{T_k}{T}\right)^{2\beta}+1+(T-T_k)\int_{\frac{T_k}{T}}^1x^{2\beta}dx \\ \label{eq:439}
&<1+\frac{T'}{1+2\beta}\leq\frac{2T'}{1+2\beta}, 
\end{align}
where Eq \eqref{eq:439} holds because $T\geq T_k+1+2\beta$. 

Consequently, a sufficient condition for Eq \eqref{eq:T0condition0} to hold is 
\begin{equation*}
\frac{\Delta^2}{16}\left(\frac{8T'}{(8+\Delta)(1+\beta )}\right)^2 > 2 (2+\frac{\Delta}{2})^2\frac{2T'}{1+2\beta}\log \frac{1}{\delta},
\end{equation*}
which yields
\begin{equation}\label{eq:T0condition1}
T' \geq \frac{(1+\beta)^2(4+\Delta)^2(8+\Delta)^2 }{4(1+2\beta)\Delta^2}\log \frac{1}{\delta}.
\end{equation}
Since Eq \eqref{eq:T0condition1} implies $T'>1+2\beta$, we can simply take
\begin{equation}
T_0=\max\Big\{\frac{(1+\beta)^2(4+\Delta)^2(8+\Delta)^2 }{4(1+2\beta)\Delta^2}\log \frac{1}{\delta}, \left[(1+\frac{16}{\Delta})^{\frac{1}{1+\beta}}-1\right]\cdot T_k \Big\}.
\end{equation}
Hence, we complete the proof.
\end{proof}

\section{ Elimination Length in Akerlof's Market for ``Lemons''}\label{appendix:sec:lemon}

\begin{proposition} \label{prop:lemon}
Suppose each seller observes his exact car quality, i.e., $\tilde{q}_i=q_i$. With any $c_1>0, c_2 > 1$, the Market for ``Lemons'' game has elimination length $L_0$ at least $2\lceil\frac{N}{k}\rceil-1$ if $q_i - q_{i-k} \geq c_1 > q_i - q_{i-k+1}, \forall i\in \{k+1, k+2, \cdots, N\}$ and $\mathcal{P} \supseteq \{q_1, q_2, \dots, q_N\} $. The buyer offering any price $p \leq q_1$ and each seller $i$ setting $a_i = 0$ are Nash equilibria of the game.
\end{proposition}

\begin{proof}
In this proof, we say a seller $i$ remains on the market if his $a_i=1$ is not eliminated. We start with the following two claims about the dominance elimination by the buyer and sellers.
\begin{claim}\label{claim:seller-eliminate}
 A seller with quality $q_i$ should remain on the market if and only if the buyer have not eliminated all its action $ p \geq q_i + c_1$. 
 Meanwhile, no seller $i$ should eliminate his $a_i=0$, as long as the buyer has not eliminate its $p=q_1$.
\end{claim}
\begin{proof}
By construction, if a seller choose not to list, his utility is always $0$ in regardless of the action of other sellers or buyer.

$(\Rightarrow)$: With some $p\geq q_i+c_1 $, the seller would have non-negative utility, $u_i= p - q_i - c_1 \geq 0$ if he choose to list. This is no worse than the zero utility if he does not list.

$(\Leftarrow)$: For any $p < q_i + c_1$, the seller with quality $q_i$ have utility $ u_i= \min(p - q_i, 0) - c_1 < 0  $, always negative if he choose to list. Since this utility is strictly dominated by than the zero utility if he does not list, thus $a_i=1$ should be eliminated.

Meanwhile, if the buyer sets a price $p = q_1$, any seller who list his car would incur a negative utility $-c_i$ strictly worse than the zero utility of not listing. So in this case no seller should always list its car.
\end{proof}

\begin{claim}\label{claim:buyer-eliminate}
The buyer should eliminate all its action $p > q$ if the highest car quality of the remaining sellers on the market is $q$. Conversely, as long as a seller of quality $q_i$ remains on the market, the buyer should not eliminate its action $p=q_i$.
\end{claim}
\begin{proof}
$(\Leftarrow)$: Among the remaining sellers on the market, let the highest car quality be $q$. Consider the buyer sets some price $p > q$, cf. $p= q$. The outcome from the two different price are the same in the sense that the buyer can get all the cars that the sellers choose to list with quality average quality $\bar{q}$. We can see that the buyer's revenue $c_2 \bar{q}$ is the same, but she has strictly minimum cost at $p=q$, therefore all the price $p > q$ are actions dominated by $p=q$ . 

$(\Rightarrow)$: Consider the situation that only the seller $i$ choose to list. In this case, the best response of the buyer is to set her price at $q_i$. This is because she does not want to set a price above $q_i$, according to the argument in paragraph above; setting a price below $q_i$, the seller will not sell and her utility is $0$, which is strictly worse than her utility $(c_2-1) q_i > 0$ if the buyer sets her price $p=q_i$. 
\end{proof}

We now provide an induction argument for the iterative elimination: 

\paragraph{Base case:} In the beginning of elimination, we know for any seller $i$, $i \leq N-k$, has his quality $q_i \leq q_N - c_1$. They cannot eliminate any of their actions as the buyer have not eliminate its action $p=q_1$, or $p=q_N \geq q_i + c_1$, according to Claim \ref{claim:seller-eliminate}. Meanwhile, we show that it takes at least $1$ round to eliminate both buyer's action(s) $p > q_{N}$ and the sellers' action $a_{i}=1, \forall i>N-k$. This is because the buyer may or may not need first eliminate all her action $p > q_N$ depending on the support of $\mathcal{P}$, according to Claim \ref{claim:buyer-eliminate}; In the same or the following round (also depending on the support of $\mathcal{P}$), all seller $i> N-k$ must eliminate their action $a_i = 1$, according to Claim \ref{claim:seller-eliminate}.



\paragraph{Inductive case:} We show the following inductive statement: for $i \in [1, \dots, \lceil\frac{N}{k}\rceil-1  ]$, at the beginning of round $2i+1$, given that we have eliminated $\{ a_{j}=1 | \forall j> N-ik \} \cup \{p > q_{N-(i-1)k} \}$,  then it takes two rounds to first eliminate all of the buyer's action(s) $p > q_{N-ik}$ and subsequently the all seller $j$'s action $a_{j}= 1$ with $j>N-(i+1)k$.

Given that we have eliminated $\{ a_{j}=1 | \forall j> N-ik \} \cup \{p > q_{N-(i-1)k} \}$, in the first round, we can eliminate of the buyer's actions $p > q_{N-ik}$, according to Claim~\ref{claim:buyer-eliminate}, as the highest car quality of the remaining sellers on the market is $q_{N-ik}$. We cannot eliminate the action of any seller $j$ with quality $q_{j} \leq q_{N-ik} \leq q_{N-(i-1)k} - c_1$, according to Claim~\ref{claim:seller-eliminate}, as the buyer could offer a price $p=q_1$, or  $p=q_{N-(i-1)k} \geq q_j + c_1$. 

In the second round, we can eliminate all seller $j$'s action $a_{j}= 1$ with $j>N-(i+1)k$, according to Claim~\ref{claim:seller-eliminate}, as the buyer eliminates all of her actions above $q_{N-ik}$. Yet we still cannot eliminate the action of any other seller, according to Claim \ref{claim:seller-eliminate}, as the buyer could offer a price $p=q_1$, or $p=q_{N-ik} \geq q_j + c_1, \forall i \leq N-(i+1)k$. We also cannot eliminate any action of the buyer, as no seller exit the market in the last round. 

Therefore, at the beginning of round $2i+3$, we expand the elimination set to $\{ a_{j}=1 | \forall j > N-(i+1)k \} \cup \{p > q_{N-ik} \}$, so the induction follows. It is easy to see that in the last iteration when $N-\lceil\frac{N}{k}\rceil k < 0$, all sellers exits the market, and all buyer's price above $q_{\nu}$ are eliminated, where we can compute $\nu=N+k-\lceil\frac{N}{k}\rceil k = \begin{cases} N \mod k, & k \nmid N \\  1, & k \mid N \end{cases} \geq 1$.
This terminates the iterative process of dominance elimination. The remaining action profiles must include all Nash equilibrium. Since all sellers have only one action left, their remaining action profile, i.e., not to list their cars, is the Nash equilibrium strategy. The buyer has utility zero for any of her actions $p\leq q_1 \leq q_\nu$, which are best responses to the sellers' equilibrium strategies and therefore form Nash equilibria along with the sellers' unique action profile.

As the base case takes at least one round to reach, and the induction stage takes $2\lceil\frac{N}{k}\rceil-2$ round, in total, the elimination length is at least $2\lceil\frac{N}{k}\rceil - 1$.
\end{proof} 
 

\end{document}